\documentclass[conference,letterpaper]{IEEEtran}
\IEEEoverridecommandlockouts

\usepackage{amsmath,amssymb,amsfonts,amsthm}
\usepackage[linesnumbered,ruled,vlined]{algorithm2e}
\usepackage{bm}
\usepackage{graphicx}
\usepackage[caption=false]{subfig}
\usepackage{enumitem}
\usepackage{adjustbox}
\usepackage{cite}
\usepackage{url}
\usepackage{dsfont}
\usepackage{setspace}
\usepackage{bbm}
\usepackage{color}
\usepackage{authblk}

\def\BibTeX{{\rm B\kern-.05em{\sc i\kern-.025em b}\kern-.08em
		T\kern-.1667em\lower.7ex\hbox{E}\kern-.125emX}}

\newcommand{\Ex}{\mathbb{E}}
\newcommand{\pr}{\mathbb{P}}

\newcommand{\rt}{\right}
\newcommand{\lt}{\left}

\newcommand{\E}{\mathcal{E}}

\newcommand{\PP}{\mathcal{P}}

\newcommand{\indi}{\mathbbm{1}}

\newtheorem{theorem}{Theorem}

\newtheorem{corollary}{Corollary}

\allowdisplaybreaks

\begin{document}

\title{Controlling Epidemic Spread Under Immunization Delay Constraints\vspace{-0.1in}\thanks{This work was supported in part by the National Science Foundation under Grant Nos. 1910749, 2007423, 2209921, and 2209922. Part of this work was done while Shiju Li was with Texas State University.}}

\author[$\ast$]{Shiju Li}
\author[$\dagger$]{Xin Huang}
\author[$\dagger$]{Chul-Ho Lee}
\author[$\ddagger$]{Do Young Eun\vspace{-0.1in}}

\affil[$\ast$]{Florida Institute of Technology, \textsuperscript{$\dagger$}Texas State University, \textsuperscript{$\ddagger$}North Carolina State University\vspace{-0.1in}}

\maketitle

\begin{abstract}
In this paper, we study the problem of minimizing the spread of a viral epidemic when immunization takes a non-negligible amount of time to take into effect. Specifically, our problem is to determine which set of nodes to be vaccinated when vaccines take a random amount of time in order to maximize the total reward, which is the expected number of saved nodes. We first provide a mathematical analysis for the reward function of vaccinating an arbitrary number of nodes when there is a single source of infection. While it is infeasible to obtain the optimal solution analytically due to the combinatorial nature of the problem, we establish that the problem is a monotone submodular maximization problem and develop a greedy algorithm that achieves a $(1\!-\!1/e)$-approximation. We further extend the scenario to the ones with multiple infection sources and discuss how the greedy algorithm can be applied systematically for the multiple-source scenarios. We finally present extensive simulation results to demonstrate the superiority of our greedy algorithm over other baseline vaccination strategies.
\end{abstract}

\section{Introduction}

Epidemic models have been important mathematical tools in analyzing the spread and control of epidemics with various applications, ranging from infectious diseases to malware propagation. Their importance has been evident in the COVID-19 crisis. The spread of infectious diseases has been popularly modeled via the compartment and metapopulation models~\cite{Daley99a,Newman10}, where each individual has equal chance to contact others in the entire population or each subpopulation. On the other hand, epidemic models on networks have been more actively used to model malware propagation over networks~\cite{Mieghem-ToN09,Newman10,Nowzari-CS16,lee2019transient,li2020trapping}. The Internet-connected devices are always exposed and vulnerable to malware and worm attacks via not only their underlying networks but also their users' social networks.

\vspace{2pt}
\noindent \textbf{Prior Work and Motivation:} Most studies on epidemic models on networks have been concerned about the \emph{persistence} and \emph{extinction} of the epidemics in their steady state. In other words, their central question is under what conditions an epidemic dies out quickly or lasts for a long period of time, or when the desired steady state of extinction can be achieved. Earlier studies established the \emph{epidemic threshold}, below which the epidemic dies out eventually over time, for susceptible-infected-susceptible/removed (SIS/SIR) models and similar variants~\cite{Chakrabarti08,Ganesh05,Draief-AAP08,Prakash-ICDM11}. It has then become the fundamental basis for the development of countermeasures or vaccination policies by manipulating the underlying network structure, e.g., removing $k$ nodes or edges~\cite{Tong-TKDE16,Tong-TKDD16,Vullikanti-SDM15}, or controlling epidemic parameters~\cite{Nowzari-CS16,Preciado-TCNS14} to achieve the `below-the-threshold' condition for the extinction of the epidemic.

While such a steady-state analysis has been crucial in epidemic modeling and control, it is also important to understand the transient dynamics of epidemic spreading over a network. In particular, when it comes to malware propagation, there is a non-negligible amount of time for a patch or vaccine to become available after the outbreak of an epidemic (virus spread), during which only infections take place over the network. It often remains unknown for a number of days that such an attack has occurred, as is the case with `zero-day' attacks~\cite{Bilge12}. This observation has led to a recent study on characterizing the transient dynamics of the susceptible-infected (SI) model on a network~\cite{lee2019transient}.

Despite the abundant literature on epidemic modeling and control, there is  still an important yet overlooked component in the development of vaccination policies. It is the presence of another non-negligible amount of time for  vaccines to come into effect, \emph{even after} the vaccines are available. While nodes are vaccinated, they are not immediately immune, but they are fully vaccinated after some amount of time. That is, the nodes who are vaccinated yet in such an immunization process are still vulnerable to infection. For example, the software patching process often undergoes multiple rounds of software installations with a possible failure in each round, which lead to a non-negligible delay in the patching/vaccination process~\cite{Nesara22CSCW,Nesara22,Wang17}. In 2019, Google's Project Zero~\cite{Google21} publishes their tracking records for publicly known cases of detected zero-day exploits, and the records indicate that it takes 15 days on average for vendors to patch a vulnerability that is being used in active attacks. A recent study reveals that the average time for companies to patch a vulnerability, or a CVE (Common Vulnerability and Exposures), is 215 days~\cite{Hacker23}. In addition, the two-dose vaccines against COVID-19 take two weeks to fully kick in after the second dose~\cite{cdc21}.

\vspace{1pt}
\noindent \textbf{Our Contributions:} In this paper, we study the problem of controlling the spread of a viral epidemic on a network when vaccines take a non-negligible amount of time to take into effect. Specifically, the problem is to determine which set of nodes in a network to be vaccinated to maximize the total reward, which is the expected number of saved nodes, when vaccines take a (common) random amount of `\emph{immunization time}' to take into effect. The epidemic spreading process here is governed by the SI model, where there is no curing process for infected nodes, while the healthy nodes who are vaccinated and remain healthy for the immunization time eventually become immune and non-infectious to others. In other words, the vaccinated nodes can still be infected during the immunization time.

We focus on the optimization problem under arbitrary tree networks for rigorous mathematical analysis. It has a higher reward to `save' a node with more children and descendants, which prevents infection from getting through, but such a node can also be at higher risk of infection as it tends to be closer to the source(s) of infection. In other words, the seemingly important nodes if the vaccination takes effect immediately may not be so important as it takes some non-negligible amount of time for vaccines to be effective, in which case they are more likely to end up getting infected before they are fully immune. The characterization of the reward and the risk of infection under a limited vaccine budget makes the analysis non-trivial, as shall be shown in Sections~\ref{se:results} and \ref{se:results2}. In addition, it is worth noting that tree networks have been similarly used in the literature for a relevant yet different problem, which is to localize the source of a viral epidemic, initially under the SI model with a single source~\cite{Shah2011rumor,shah2012rumor} and then extended to other models and scenarios~\cite{luo2013identifying,ying2016rumor,choi2019}.

The main contributions of this work are as follows:
\vspace{-2pt}
\begin{itemize}[itemsep=2pt,leftmargin=1.2em]
\item First, we formulate the problem of minimizing the spread of a viral epidemic, i.e., maximizing the number of saved nodes (total reward), under the immunization delay and limited vaccination budget constraints.

\item Second, we provide a mathematical analysis for the reward function of vaccinating a given number of nodes for the case of a single infection source. We then prove that the optimization problem is a monotone submodular maximization problem and develop a greedy algorithm that achieves a $(1\!-\!1/e)$-approximation.

\item Third, we discuss how this greedy algorithm can be applied to the scenarios with multiple infection sources.

\item Finally, we present extensive simulation results to demonstrate the superiority of our greedy algorithm over baseline vaccination strategies under a wide range of scenarios.
\end{itemize}

\section{Notation and Problem Formulation}\label{se:prelim}

Consider a connected tree graph $G \!=\! (V, E)$ with $|V| \!=\! n$. Let $s \in V$ be the source of a viral epidemic that starts at time 0, which is assumed to be the root node of the tree graph $G$. The epidemic spread is governed by the SI model in that an infected node can infect each of its neighbors (children) with infection rate $\lambda$. In other words, an infection from node $i$ to its child $j$ takes place after an exponential amount of time with mean $1/\lambda$. Let $d_i$ be the depth of node $i$, and let $Z_i$ be the time until node $i$ is infected. Then, we see that
$Z_j \!=\! \sum^{d_i}_{j = 1} X_j$, where $X_i$ is an \textit{i.i.d.} copy of the exponential random variable with rate $\lambda$. Also, let $A_i$ be the set of ancestors of node $i$, including its parent, and let $N_i$ be the set of descendants of node $i$ with $|N_i| \!=\! n_i$.

Let $k$ be the number of available vaccines, or the vaccination budget. $k$ vaccinated nodes take a random amount of immunization time, denoted by $\tau$, to be fully immune and no longer infectious to others, given that they remain healthy during the immunization time. In other words, vaccinated nodes are still vulnerable to infection and can be infected during the immunization time. Then, we see that the probability that node $i$, when vaccinated, eventually becomes immune is $\pr\{Z_i\!>\!\tau\}$, which can be obtained by conditioning on $Z_i$ (or $\tau$). That is, we have
\begin{equation*}
	\pr\{Z_i>\tau\} = \Ex\lt[\pr\{Z_i>\tau|Z_i\}\rt].
\end{equation*}
For example, if $\tau$ is exponentially distributed with rate $\mu$, we have
\begin{equation*}
	\pr\{Z_i > \tau\} = 1 - \Ex\lt[e^{-\mu Z_i}\rt]= 1- \lt(\frac{\lambda}{\lambda+\mu}\rt)^{d_i},
\end{equation*}
where the expectation is with respect to $Z_i$, and the second equality is from the Laplace transform of $Z_i \!=\! \sum^{d_i}_{j = 1} X_j$. We also define $I_i$ to be an indicator variable that node $i$ becomes immune when it is vaccinated, which is given by
\begin{equation*}
  I_i := \mathds{1}\{Z_i > \tau\},
\end{equation*}
with $\Ex[I_i] \!=\! \pr\{Z_i\!>\!\tau\}$. We collect all the notations in Table~\ref{table1}. See Figure~\ref{fig:prelim} for an illustration, where $d_i \!=\! 2$, $n_i \!=\! 4$, and $Z_i \!=\! X_1 \!+\! X_2$, which follows the so-called Erlang-2 distribution.

\begin{table}[t]
	\renewcommand{\arraystretch}{1.1}
	\caption{Notations} \label{table1}
	\vspace{-4mm}
	\begin{center}
		\centering
		\small
		\begin{tabular}{|c|c|c|c|}
			\hline
			$\lambda$ & Infection rate \\
			\hline
			$\tau$ & Immunization time for vaccines to take into effect \\
			\hline
			$k$ & Vaccination budget \\
			\hline
			$d_i$ & Depth of node $i$ \\
			\hline
			$A_i$ & Set of ancestors of node $i$  \\
			\hline
			$N_i$ & Set of descendants of node $i$  \\
			\hline
			$n_i$ & Number of descendants of node $i$, i.e., $n_i=|N_i|$   \\
			\hline
			$X_j$ & \textit{i.i.d} copy of the exponential random variable with $\lambda$ \\
			\hline
			$Z_i$ & Time until node $i$ is infected, i.e., $Z_i=\sum_{j=1}^{d_i}X_j$ \\
			\hline
			$I_i$ & Indicator variable that node $i$ is immune after vaccination  \\
			\hline
		\end{tabular}
	\end{center}
	\vspace{-2mm}
\end{table}

\begin{figure}[t]
	\centering
	\vspace{-2mm}
	\includegraphics[width=0.5\linewidth, trim=0.3mm 0.3mm 0.3mm 0mm, clip]{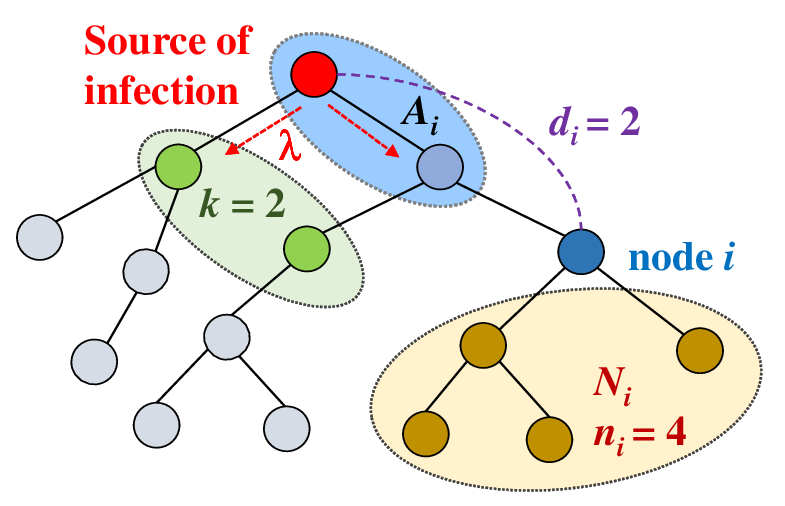}
	\vspace{-2mm}
	\caption{An example to illustrate the notations.}
	\label{fig:prelim}
	\vspace{-5mm}
\end{figure}

For a given budget $k$, let $r(S)$ denote the expected total reward, i.e., the expected number of saved nodes, by vaccinating a node set $S$, where $|S| \!=\! k$. Our problem is then to find which $k$ nodes in $G$ to be vaccinated to maximize the expected total reward under the immunization delay constraint, which is formally given as the following optimization problem:
\begin{align*}\label{Problem}
	\bm{\mathcal{P}:} &&  S^\star = \underset{S \subseteq V ~:~ |S|= k}{\arg\max} r(S). &&
\end{align*}

The first step to solving this problem is to characterize $r(S)$ in a closed form. Consider $k \!=\! 1$. Let $R_i$ be the reward of vaccinating node $i$, which is the total number of saved nodes when node $i$ is vaccinated and immune after the immunization time $\tau$. Then, we see that $R_i \!=\! n_i I_i$, and
\begin{equation}\label{k1}
  \Ex[R_i] = n_i \Ex[I_i] = n_i \pr\{Z_i > \tau\}.
\end{equation}
In other words, if node $i$ is immune after $\tau$, it is no longer infectious to its children, thus preventing the epidemic from getting through its descendants (including the children). Here the number of its descendants is $n_i$. Thus, we have the expected total reward $r(\{i\}) \!=\! \Ex[R_i]$. While it is straightforward to characterize $r(S)$ for $k \!=\! 1$ as we have shown, it quickly becomes \emph{non-trivial} as the value of $k$ increases. It is because the expected total reward for $k \!>\! 1$ cannot be simply written as the sum of the expected rewards of vaccinating each individual node assuming that $k \!=\! 1$.

\section{Main Results}\label{se:results}

In this section, we first characterize $r(S)$ for $k \!=\! 2$. We then obtain a general expression for $r(S)$ for any $k \!\geq\! 1$.  We next establish that $\bm{\PP}$ is a monotone submodular maximization problem and finally propose a $(1\!-\!1/e)$-approximation greedy algorithm for solving $\bm{\PP}$.

\subsection{Reward Function for $k \!=\! 2$}

Fix $S \!=\! \{i,j\}$. Let $R_S$ be the total reward when nodes $i$ and $j$, or the node set $S$, are chosen to be vaccinated. Note that $r(S) \!=\! \Ex[R_S]$. To characterize $R_S$ and $r(S)$, we can consider the following two general cases based on where the two nodes are on a tree.

\vspace{1mm}
\noindent \textbf{Case (1): $N_i\cap N_j \!=\! \emptyset$}. This is the case when the set of descendants of $i$, $N_i$, does not intersect the set of $j$, $N_j$. Since $N_i$ and $N_j$ are the saved nodes when $i$ and $j$ are immunized, respectively, the rewards from immunizing $i$ and $j$ can be computed separately. They are just $n_i$ and $n_j$, respectively. We say that $i$ and $j$ are \emph{collateral} in this case. Thus, we have
\begin{equation*}
\begin{split}
  R_S &= n_i\mathds{1}\{I_i=1, I_j=0 \}+n_j\mathds{1}\{I_i=0, I_j=1\} \\
   &~~~+ (n_i+n_j)\mathds{1}\{I_i=1,I_j=1\}.
\end{split}
\end{equation*}
Taking expectations yields
\begin{align}
r(S) &= n_i\pr\{Z_i\!>\!\tau,Z_j\!<\!\tau\} + n_j\pr\{Z_i\!<\!\tau,Z_j\!>\!\tau\} \nonumber\\
		&~~~+(n_i+n_j)\pr\{Z_i>\tau,Z_j>\tau\} \nonumber\\
		&= n_i\pr\{Z_i>\tau\}+n_j\pr\{Z_j>\tau\}.\label{immediate_reward0}
\end{align}
Note that $Z_i$ and $Z_j$ may be correlated as the path from the root to $i$ and the path to $j$ may share a common edge.

\vspace{1mm}
\noindent \textbf{Case (2): $N_i\cap N_j\neq\emptyset$}. This is the case when the intersection between the set of descendants of $i$ and the one of $j$ is no longer empty. Thus, the rewards from immunizing $i$ and $j$ cannot be computed separately. We say that $i$ and $j$ are \emph{immediate} in this case. Note that $i$ and $j$ may not be immediate neighbor of each other (or they may not have a parent-child relationship). Without loss of generality, suppose that $i$ is $j$'s ancestor. Then, $i$'s descendants include all $j$'s descendants. Also, letting $d$ be the number of edges from $i$ to $j$, we see that $Z_j \!=\! Z_i \!+\! \sum_{l=1}^{d}X_l$. Thus, to compute $r(S)$, we only need to consider the following two cases, regarding whether $i$ and $j$ are finally immune after vaccination: a) $i$ is immune, and b) $i$ is infected, but $j$ is immune. In the former case, $j$ is also immune since $Z_i \!>\! \tau$ implies $Z_j \!>\! \tau$, i.e., $I_i \!=\! 1$ implies $I_j \!=\! 1$. Therefore, we have
\begin{align}
		r(S) &= n_i\pr\{I_i = 1\}+n_j\pr\{I_i = 0, I_j = 1\} \nonumber\\
		&= n_i\pr\{Z_i > \tau\}+n_j\pr\{Z_i < \tau < Z_j\} \nonumber\\
		&= (n_i-n_j)\pr\{Z_i > \tau\}+n_j\pr\{Z_j > \tau\},\label{immediate_reward}
\end{align}
which is from $\pr\{Z_i\!<\!\tau\!<\!Z_j\} \!=\! 1 \!-\! \pr\{Z_i\!>\!\tau\} \!-\! \pr\{Z_j\!<\!\tau\}$.

\subsection{General Form of Reward Function}

We next generalize the expressions of the expected total reward $r(S)$ in (\ref{immediate_reward0}) and (\ref{immediate_reward}) for $k \!=\! 2$ to the cases with $k \!\geq\! 1$. We show that $r(S)$ is in the form of a sum, where each summand is for each node $i$ in $S$. This summand involves the probability that node $i$ is finally immune after the immunization time and another term involving the reward of immunizing $i$, which is the number of its descendants, $n_i$. This second term, however, excludes the portion of $i$'s descendants that are saved by other nodes that are in $S$ and $i$'s descendants. This allows us to compute $r(S)$ correctly without double counting. Specifically, we have the following result:

\begin{theorem}\label{total_reward}
	The expected total reward of vaccinating each node in $S$ is
	\begin{equation}\label{reward_function_general}
		r(S)=\sum_{i\in S}\Big|N_i\setminus \underset{{j\in {N_i\cap S}}}{\cup}N_j\Big|\pr\{Z_i>\tau\}.
	\end{equation}	
    From the fact that $\underset{{j\in {N_i\cap S}}}{\cup} N_j \subset N_i$, we can further write (\ref{reward_function_general}) as
	\begin{equation}\label{reward_function}
		r(S)=\sum_{i\in S}\left(n_i- \Big|\underset{{j\in {N_i\cap S}}}{\cup}N_j\Big|\right)\pr\{Z_i>\tau\}.
	\end{equation}	
\end{theorem}
\begin{proof}[Proof (Sketch)]
Fix $S$. We first use the nodes in $S$ to construct one or multiple trees by creating edges between nearest `immediate' node pairs in $S$, while keeping their depth ordering. Since nodes between different trees are `collateral', the reward function for each tree can be calculated individually and added together in the end. See Figure~\ref{fig:Recur}(a) for an example, in which case there are two trees created, and the nodes in one tree are collateral to the ones in the other tree.

\begin{figure}[ht]
    \vspace{-2mm}
    \captionsetup[subfloat]{captionskip=2pt}
    \centering
    \subfloat[]{%
        \includegraphics[width=0.55\linewidth, trim=0.3mm 0.3mm 0.3mm 0mm, clip]{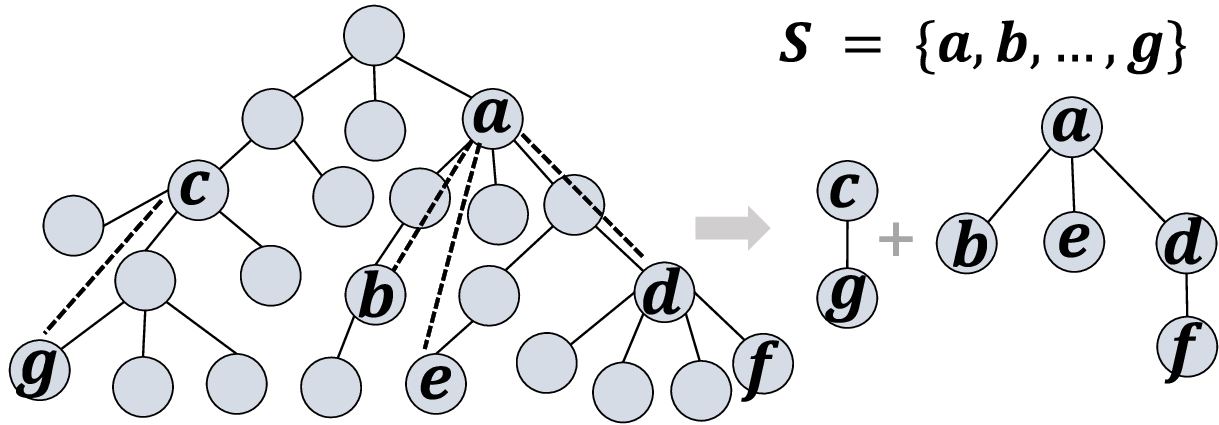}
    }\hspace{2.6mm}
    \subfloat[]{%
        \includegraphics[width=0.38\linewidth, trim=0.3mm 0.3mm 0.3mm 0mm, clip]{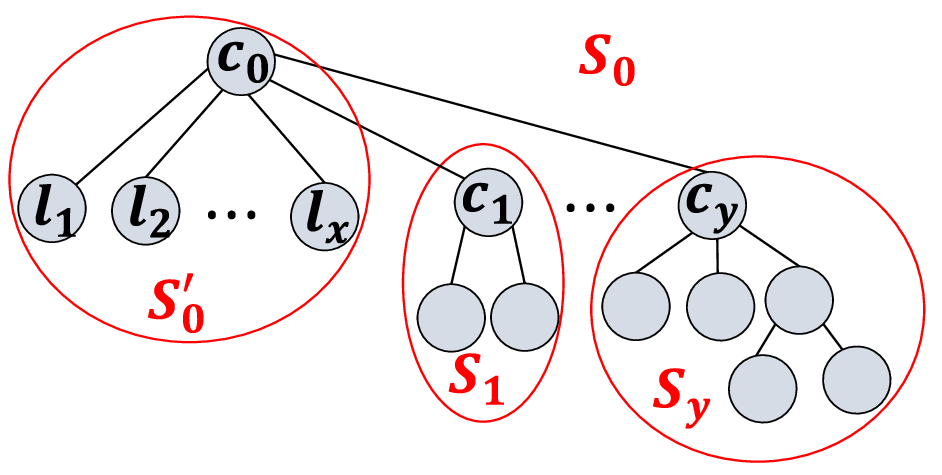}
    }\vspace{-1mm}
    \caption{Recursive structure.}
    \label{fig:Recur}
    \vspace{-2mm}
\end{figure}

We next characterize the reward function of a tree formed by a subset of $S$, say $S_0$. Due to the recursive structure of a tree, $S_0$ can be written as $S_0 := \{c_0, L, S'\}$, where $S_0$ contains root node $c_0$ and all its `children', which can be either leaf nodes $L=\{l_1,l_2,\cdots,l_x\}$ or subtrees $S'=\{S_1,S_2,\cdots,S_y\}$. See Figure~\ref{fig:Recur}(b) for an illustration. Note that each subtree has the same recursive structure as the one of $S_0$.

We start with a simpler case with $S'_0 := \{c_0, L\}$. We here characterize the reward function in this case. Let $\E_l$ be a collection of all possible combinations of the immunization states of all nodes in $L$. Note that $|\E_l|=2^x$. For node $l_i$, let $O_{l_i}$ denote that the vaccination of $l_i$ is successful, i.e., $O_{l_i}:= \{Z_{l_i}>\tau\}$. Then, we have
\begin{align*}
    &r(S'_0)=n_{c_0}\pr\{Z_{c_0}>\tau,O_{l_1},O_{l_2},\cdots,O_{l_x}\}\\
    &\quad\quad\quad + \sum_{A\in\E_l}\left(\sum^x_{i=1}{\indi}_{O_{l_i}}\cdot n_{l_i}\right)\pr\{Z_{c_0}<\tau,A\}\\
    &=n_{c_0}\pr\{Z_{c_0}>\tau\}+\sum^x_{i=1}n_{l_i}\left(\sum_{B\in\E_l\cap O_{l_i}}\pr\{Z_{c_0}<\tau,B\}\right)\\
    &=n_{c_0}\pr\{Z_{c_0}>\tau\}+\sum^x_{i=1}n_{l_i}\pr\{Z_{c_0}<\tau,O_{l_i}\}\\
    &=n_{c_0}\pr\{Z_{c_0}>\tau\}+\sum^x_{i=1}n_{l_i}\left(\pr\{Z_{l_i}>\tau\}-\pr\{Z_{c_0}>\tau\}\right)\\	&=\left(n_{c_0}-\sum^x_{i=1}n_{l_i}\right)\pr\{Z_{c_0}>\tau\}+\sum^x_{i=1}n_{l_i}\pr\{Z_{l_i}>\tau\}\\
    &=\left|N_{c_0}\setminus\underset{{j\in{N_{c_0}\cap S'_0}}}{\cup}N_j\right|\pr\{Z_{c_0}>\tau\}+\sum^x_{i=1}\left|N_{l_i}\right|\pr\{Z_{l_i}>\tau\}.
\end{align*}
where the fourth equality can be shown by using the same argument as used to obtain (\ref{immediate_reward}). It is clear that the expression of $r(S'_0)$ satisfies (\ref{reward_function_general}).
 
We can then extend the above arguments to characterize the reward function for the case of $S_0 := \{c_0, L, S'\}$ and show that $r(S)$ is given in (\ref{reward_function_general}).
\end{proof}

Consider an example network with $S \!=\! \{a,b,c,d,e\}$ in Figure~\ref{fig:3}. We can see that the reward by saving $a$, or the reward term for $a$ in (\ref{reward_function}), does not include the descendants of nodes $b$ and $c$, i.e., $N_b$ and $N_c$, respectively, to avoid double counting. Note that the descendants of node $e$, $N_e$, are already excluded as they are a subset of $N_c$. In addition, it is straightforward to see that $r(S)$ in (\ref{reward_function}) reduces to the expressions in (\ref{immediate_reward0}) and (\ref{immediate_reward}) for $k \!=\! 2$ and also to the one in (\ref{k1}) for $k \!=\! 1$.

While we can obtain a closed-form expression of $r(S)$ in (\ref{reward_function}) for any $k \!\geq\! 1$, it is infeasible to obtain the optimal solution to $\bm{\mathcal{P}}$ analytically. Since $r(S)$ depends on the underlying topology and the relationships among the nodes in $S$, we need to evaluate $r(S)$ for all possible choices of $S$. However, it is quickly infeasible to explore the search space since $\binom{n}{k} \!=\! \Theta(n^k)$ easily becomes a prohibitively large number even for a moderate $k$.

\begin{figure}[t]
	\centering
	\includegraphics[width=0.6\linewidth, trim=0.3mm 0.3mm 0.3mm 0mm, clip]{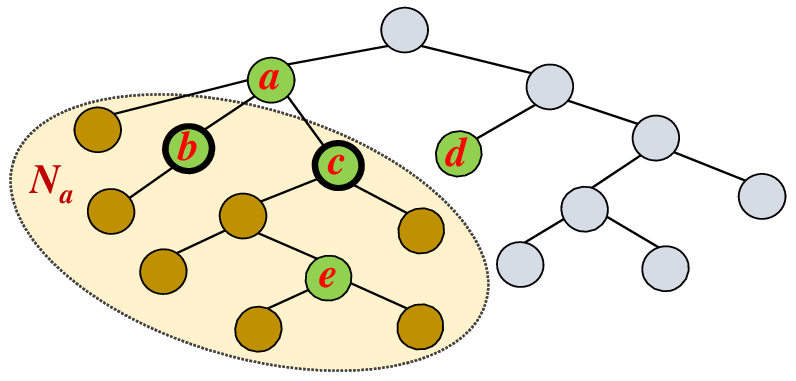}
	\vspace{-3mm}
	\caption{An example for characterizing the reward from vaccinating an arbitrary node set.}
	\label{fig:3}
    \vspace{-5mm}
\end{figure}

\subsection{Submodularity of Reward Function}

We next turn our attention to the characteristics of the set function $r$ in~(\ref{reward_function}). We below show that $r$ is non-negative, monotone, and submodular, which implies that $\bm{\mathcal{P}}$ is a monotone submodular maximization problem. While it is straightforward to see that $r$ is non-negative, it is non-trivial to show its monotonicity and submodularity.

\vspace{-1mm}
\begin{theorem}\label{submodularity}
	$r$ is a monotone submodular function.
\end{theorem}
\vspace{-3mm}
\begin{proof}[Proof (Sketch)]	
\textbf{Part I.} We here show that $r(S)$ is monotone. Letting $S\!\subseteq\! V$ and $u \!\in\! V \setminus S$, we need to show that the marginal gain is non-negative, i.e., $\Delta(S,u) := r(S\cup\{u\}) - r(S)\geq 0$, for all $S$ and $u$. Observe that
\begin{align}
&\Delta(S,u)=\sum_{i\in S\cup \{u\}}\left(n_i- \Big|\underset{j\in N_i\cap (S\cup\{u\})}{\cup}N_j\Big|\right)\pr\{Z_i>\tau\} \nonumber\\
&\quad\quad\quad\quad - \sum_{i\in S}\left(n_i- \Big|\underset{j\in N_i\cap S}{\cup}N_j\Big|\right)\pr\{Z_i>\tau\} \nonumber\\
&=\left(n_u- \Big|\underset{j\in (N_u\cap S)\cup (N_u\cap \{u\})}{\cup}N_j\Big|\right)\pr\{Z_u>\tau\} \nonumber\\
&-\sum_{i\in S}\left(\Big|\underset{j\in (N_i\cap S)\cup (N_i\cap \{u\})}{\cup}N_j\Big| - \Big|\underset{j\in N_i\cap S}{\cup}N_j\Big|\right)\pr\{Z_i \!>\! \tau\}. \label{sss}
\end{align}

\noindent (i) $S\cap A_{u}=\emptyset$: We have $N_i\cap \{u\}=\emptyset$ for all $i\in S$, which implies that $u$ is not $i$'s descendant. Then, we have
\begin{equation*}	
\Big|\underset{{j\in {(N_i\cap S)\cup (N_i\cap \{u\}) }}}{\cup}N_j\Big|= \Big|\underset{{j\in {(N_i\cap S)\cup \emptyset }}}{\cup}N_j\Big|=\Big|\underset{{j\in {N_i\cap S}}}{\cup}N_j\Big|.
\end{equation*}
Since $N_u\cap \{u\} =\emptyset$, and from (\ref{sss}), we finally have
\begin{equation}
\Delta(S,u)=\left(n_u- \Big|\underset{{j\in {N_u\cap S }}}{\cup}N_j\Big|\right)\pr\{Z_u>\tau\}>0. \label{sub-result}
\end{equation}

\vspace{1mm}
\noindent (ii) $S\cap A_{u}\neq\emptyset$: If $N_i\cap \{u\} =\emptyset$ for all $i\in S$, then we have $\Delta(S,u) > 0$, as shown above. We thus here focus on the case when $N_i\cap \{u\} \neq\emptyset$, in which case $N_i\cap \{u\}=\{u\}$. Note that $u$ is $i$'s descendant, and $i$ is $u$'s ancestor, i.e., $i\in A_{u}$. We can then show that
\begin{align}
&\Delta(S,u) = \left(n_u- \Big|\underset{{j\in {N_u\cap S }}}{\cup}N_j\Big|\right)\pr\{Z_u>\tau\} \nonumber\\
&\quad -\sum_{i\in S\cap A_{u}}\left(n_u- \Big|N_u\cap\underset{{j\in {N_i\cap S}}}{\cup}N_j\Big|\right)\pr\{Z_i>\tau\}.\label{ttt}
\end{align}

Let $v$ be $u$'s `closest' ancestor in $S$, i.e., there is no node in $S$ that is the ancestor of $u$ and the descendant of $v$. We now consider the following two cases for the second term in (\ref{ttt}). 

\vspace{1mm}
\noindent (ii-a) Case when there exists node $i\in S\cap A_{u}$ such that $i\neq v$: Observe that $N_v\subseteq \underset{{j\in {N_i\cap S}}}{\cup}N_j$. Since $N_u\subseteq N_v$, we have $N_u\!\subseteq\! \underset{{j\in {N_i\cap S}}}{\cup}N_j$, which implies that $N_u\cap\underset{{j\in {N_i\cap S}}}{\cup}N_j=N_u$. Thus, the second term in (\ref{ttt}) becomes zero. It follows from (\ref{sub-result}) that $\Delta(S,u) > 0$.

\vspace{1mm}
\noindent (ii-b) Case when $S\cap A_{u} = \{v\}$: Observe that
\begin{align}
&N_u\cap\underset{{j\in {N_v\cap S}}}{\cup}N_j
=N_u\cap\underset{{j\in {(\{u\}\cup N_u\cup (N_v\setminus \{u\} \setminus N_u))
            \cap S}}}{\cup}N_j\nonumber\\
&\quad=N_u\cap\left(\underset{{j\in {N_u\cap S}}}{\cup}N_j\cup \underset{{j\in ({N_v\setminus \{u\}\setminus N_u})\cap S}}{\cup}N_j\right)\nonumber\\
&\quad=\left(N_u\cap\underset{{j\in {N_u\cap S}}}{\cup}N_j\right)\cup \left(N_u\cap\underset{{j\in ({N_v\setminus \{u\}\setminus N_u})\cap S}}{\cup}N_j\right) \nonumber\\
&\quad=N_u\cap\underset{{j\in {N_u\cap S}}}{\cup}N_j,\label{sub_proof}
\end{align}
where the last equality holds due to the fact that $N_u\cap\underset{{j\in ({N_v\setminus \{u\}\setminus N_u})\cap S}}{\cup}N_j=\emptyset$, since $u$ and $j$ are collateral for all $j\in {N_v\setminus \{u\} \setminus N_u}$. Then, from (\ref{sub_proof}), we can show the following: If $N_u\cap S =\emptyset$, then we have
\begin{equation*}
\Delta(S,u)=n_u\Big(\pr\{Z_u>\tau\}-\pr\{Z_v>\tau\}\Big).
\end{equation*}
Also, if $N_u\cap S \neq \emptyset$, then we have
\begin{equation}
\Delta(S,u)\!=\!\left(\!n_u\!-\!\Big|\underset{{j\in {N_u\cap S }}}{\cup}\!N_j\!\Big|\right)\!\!\Big(\pr\{Z_u\!>\!\tau\}-\pr\{Z_v\!>\!\tau\}\Big).\label{sub_conclusion}
\end{equation}
Since $n_u\geq \Big|\underset{j\in N_u}{\cup}N_j\Big|+1$, we have $n_u- \Big|\underset{{j\in {N_u\cap S }}}{\cup}N_j\Big|>0$.
Also, since $v$ is $u$'s ancestor, we have $Z_v<Z_u$ almost surely, thereby leading to $\pr\{Z_u \!>\! \tau\} \!>\! \pr\{Z_v \!>\! \tau\}$. Thus, we have $\Delta(S,u)>0$.

Therefore, from (i) and (ii), it follows that $\Delta(S,u)>0$.

\vspace{1mm}
\noindent \textbf{Part II.} We next show that $r(S)$ is submodular. Fix $S_1\subseteq S_2 \subseteq V$ and $u \in V \setminus S_2$. We need to show that $\Delta(S_1,u) \geq \Delta(S_2,u)$. To this end, we need to consider the following three cases:

\vspace{1mm}
\noindent (i) $S_1\cap A_{u}\neq \emptyset$ and $S_2\cap A_{u}\neq \emptyset$: Let $v_1$ be $u$'s closest ancestor in $S_1$, and let $v_2$ be $u$'s closest ancestor in $S_2$. By following a similar argument as used to derive (\ref{sub_conclusion}), we can show that
\begin{equation*}
\Delta(S_1,u)\!\!=\!\!\left(n_u\!-\!\Big|\underset{{j\in {N_u\cap S_1 }}}{\cup}N_j\Big|\right)\!\!\Big(\pr\{Z_u>\tau\}\!-\!\pr\{Z_{v_1}>\tau\}\Big).
\end{equation*}
Similarly for $\Delta(S_2,u)$. We can see that $d_{v_1}\leq d_{v_2}$ due to the possible closer ancestor of $u$ than $v_1$ from $S_2\setminus S_1$. Since $v_1$ and $v_2$ are both $u$'s ancestors, they are immediate. Thus, we have $Z_{v_1}<Z_{v_2}$ almost surely, thus implying that $\pr\{Z_{v_2}>\tau\}>\pr\{Z_{v_1}>\tau\}$. In addition, since $\underset{{j\in {N_u\cap S_2 }}}{\cup}N_j=\{\underset{{j\in {N_u\cap S_1 }}}{\cup}N_j\}\cup \{\underset{{j\in {N_u\cap \{S_2\setminus S_1\} }}}{\cup}N_j\}$, we have $\underset{{j\in {N_u\cap S_1 }}}{\cup}N_j\subseteq\underset{{j\in {N_u\cap S_2 }}}{\cup}N_j$. Thus, we have $n_u- \Big|\underset{{j\in {N_u\cap S_1 }}}{\cup}N_j\Big|\geq n_u- \Big|\underset{{j\in {N_u\cap S_2 }}}{\cup}N_j\Big|$. As a result, $\Delta(S_1,u) \geq \Delta(S_2,u)$.
	
\vspace{1mm} 
\noindent (ii) $S_1\cap A_{u}= \emptyset$ and $S_2\cap A_{u}=\emptyset$: We can show that
\begin{equation*}
\Delta(S_1,u) = \left(n_u- \Big|\underset{{j\in {N_u\cap S_1 }}}{\cup}N_j\Big|\right)\pr\{Z_u>\tau\}.
\end{equation*}
Similarly for $\Delta(S_2,u)$. By using a similar argument as the one for case (i), we have $\Delta(S_1,u) \geq \Delta(S_2,u)$.

\vspace{1mm} 
\noindent (iii) $S_1\cap A_{u}= \emptyset$ and $S_2\cap A_{u}\neq\emptyset$: We can show that
\begin{align*}
&\Delta(S_1,u)=\left(n_u- \Big|\underset{{j\in {N_u\cap S_1 }}}{\cup}N_j\Big|\right)\pr\{Z_u>\tau\},\\
&\Delta(S_2,u)\!=\!\left(n_u\!-\!\Big|\underset{{j\in {N_u\cap S_2 }}}{\cup}N_j\Big|\right)\!\!\Big(\pr\{Z_u\!>\!\tau\}\!-\!\pr\{Z_{v_2}\!>\!\tau\}\Big).
\end{align*}
By using a similar argument as the one for case (i), we have $\Delta(S_1,u)> \Delta(S_2,u)$. Therefore, we complete the proof of  $\Delta(S_1,u) \geq \Delta(S_2,u)$. 
\end{proof}
\vspace{-2mm}

\subsection{Greedy Algorithm}
Thanks to the nice properties of $r$, we are able to develop a greedy algorithm for solving $\bm{\mathcal{P}}$, which is summarized in Algorithm~\ref{algo}. For any given budget $k$, it is, in essence, to find node $u^*$ that maximizes the marginal gain $\Delta(S,u)$ every iteration until the size of the resulting set $S$ becomes $k$. Therefore, this algorithm naturally achieves a $(1\!-\!1/e)$-approximation performance guarantee. While we refer to the proof of Theorem~\ref{submodularity} for more details, we note that $\Delta(S,u)$ has two possible expressions, which are given in Line 9 and Line 12 in Algorithm~\ref{algo}, respectively. An expression is chosen depending on whether the (growing) set $S$ and the set of ancestors of $u$, $A_u$, have a non-empty intersection or not. Note that, in addition to the set $A_i$, the depth of node $i$, $d_i$, is also needed in Algorithm~\ref{algo}.

\setlength{\textfloatsep}{2pt}
\begin{algorithm}[t]
    \small
	\SetKwInOut{Input}{Input}
	\SetKwInOut{Output}{Output}
	\Input{$V$, budget $k$}
	\Output{$S$}
	$S \leftarrow \emptyset$\\
	\While{$|S|\leq k$}{
		\For{$u\in V \setminus S$}{
			\If{$S\cap A_u=\emptyset $}{$\Delta(S,u)=\left(n_u- \Big|\underset{{j\in {N_u\cap S }}}{\cup}N_j\Big|\right)\pr\{Z_u>\tau\}$}
			\Else{
				$v \leftarrow \underset{v\in S\cap A_u}{\arg\max}\;d_v $\\
				$\Delta(S,u)=\left(n_u- \Big|\underset{{j\in {N_u\cap S }}}{\cup}N_j\Big|\right)\Bigl(\pr\{Z_u >\tau\}\allowbreak -\pr\{Z_v>\tau\}\Bigr)$}
		}
		$u^* \leftarrow \underset{u\in V \setminus S}{\arg\max}\;\Delta(S,u) $\\
		$S \leftarrow S \cup \{u^*\}$ \\
	}
	\caption{Greedy Algorithm}\label{algo}
\end{algorithm}

\begin{corollary}
	Our greedy algorithm in Algorithm~\ref{algo} achieves a $(1\!-\!1/e)$-approximation to the optimal solution of $\bm{\PP}$.
\end{corollary}
\vspace{-3mm}
\begin{proof}
This result follows from Theorem~\ref{submodularity} and~\cite{nemhauser1978analysis}.
\end{proof}

\section{Discussions on Multiple-Source Scenarios}\label{se:results2}

In this section, we discuss how our greedy algorithm can be applied to the scenarios with multiple infection sources. We observe that with a proper transformation of the underlying tree structure, the cases with multiple sources can boil down to three fundamental cases.

Before going into the details of each fundamental case, we consider the scenario where the infection source is not the root node of a tree. In this case, we can easily transform the tree into another one having the source of infection as the root while updating the parameters based on the structure of the transformed one. See Figure~\ref{fig:variant} for an illustration on this transformation. Thus, our discussion below is done based on the (transformed) tree whose root is a source of infection in addition to other sources.

\subsection{Disconnected Trees with Different Sources}\label{9.1}

We first consider the case where there are two disconnected trees, each of whose root is the source of infection for each tree. The propagation of infection over each tree is also governed by the SI model with the same infection rate $\lambda$. Let $G_i \!=\! (V_i, E_i)$, $i \!=\! 1,2$, denote each disconnected tree. We can then write the expected total reward $r(S)$ as follows:
\begin{equation}\label{xxxx}
	r(S)=r_{G_1}(S_1)+r_{G_2}(S_2)=r_{G_1}(S)+r_{G_2}(S),
\end{equation}
where $S\!=\!S_1 \cup S_2$, $S_1 \!\subseteq\! V_1$, and $S_2 \!\subseteq\! V_2$. From Theorem~\ref{submodularity}, we know that $r_{G_1}(S_1)$ is submodular. Since a sum of two submodular functions is also submodular, $r(S)$ in (\ref{xxxx}) is a submodular function. Similarly for the monotonicity. Thus, we can use Algorithm~\ref{algo} as follows.

In the first iteration, we run Algorithm~\ref{algo} on $G_1$ and $G_2$ independently to find candidate nodes $u^*_1$ and $u^*_2$ along with their corresponding maximum marginal gains $\Delta(S_1,u^*_1)$ and $\Delta(S_2,u^*_2)$, respectively. We then compare the values of $\Delta(S_1,u^*_1)$ and $\Delta(S_2,u^*_2)$, and add $u^*$ that leads to a larger gain into the node set $S$ for vaccination. We next run Algorithm~\ref{algo} on the tree where $u^*$ comes from, to find a new candidate node. We compare this new candidate node with the one with a smaller gain in the previous iteration. Again, the candidate node that gives a larger marginal gain in this iteration is added into $S$. We repeat the above process until the size of $S$ reaches $k$.

This argument can be readily extended to the cases with three or more disconnected trees having their roots as the infection sources. In other words, we just need to run Algorithm~\ref{algo} on each tree individually to update a candidate node and select the one for $S$ that leads to the largest marginal gain across the trees in each iteration. This process is repeated until the size of $S$ equals $k$.

\begin{figure}[t]
	\centering
	\includegraphics[width=0.8\linewidth, trim=1mm 1mm 1mm 0mm, clip]{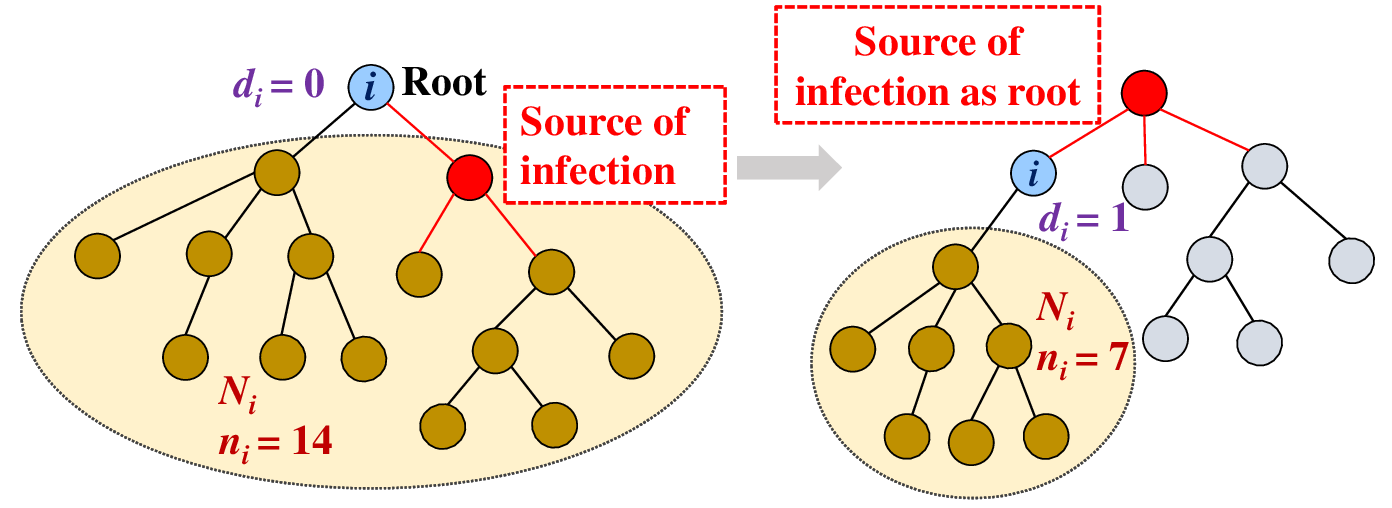}
	\vspace{-2mm}
	\caption{Tree structure transformation.}
	\label{fig:variant}
    \vspace{1mm}
\end{figure}

\subsection{Connected Multiple Sources}\label{9.2}

The next fundamental case is when there are multiple initially infected nodes and they form a (sub)tree. This case would correspond to the following scenario:  The spread of a viral epidemic starts from the root, and it lasts for a certain amount of time, so there are multiple infected nodes. Since then, $k$ vaccines are available. In other words, the `vaccine intervention' takes place after some initial delay (until vaccines are available). Note that there still exists the immunization time $\tau$ after which the vaccines come into effect. 

Observe that the epidemic continues to spread over the tree from `frontier' infected nodes, which are the \emph{leaf} nodes of a tree made by the initially infected nodes. Let $G'$ be the tree by the initially infected nodes (multiple infection sources). We first remove all the edges in $G'$ from the original tree $G$, which are the edges among the initially infected nodes. Then, $G$ can be split into a set of trees in each of which the root is a source of infection while other nodes, if any, are healthy nodes. We provide an example in Figure~\ref{fig:Multiple}, where the highlighted edges in $G'$ are removed, leading to three (sub)trees $G_1$, $G_2$, and $G_3$. Hence, this second case now boils down to the case in Section~\ref{9.1}, in which there are multiple disconnected trees, each of whose root is the (only) source of infection for each tree. We can then apply our greedy algorithm in Algorithm~\ref{algo} as explained above.

\begin{figure}[t]
	\centering
	\includegraphics[width=0.4\linewidth, trim=1mm 1mm 1mm 0mm, clip]{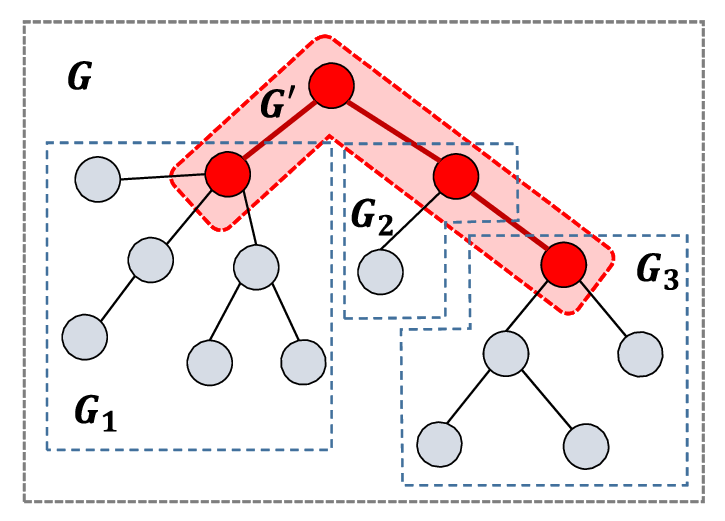}
	\vspace{-2mm}
	\caption{Infected nodes (sources) form a tree.}
	\label{fig:Multiple}
	\vspace{-3mm}
\end{figure}

\begin{figure}[t]
	\centering
	\includegraphics[width=1\linewidth, trim=1mm 1mm 1mm 0mm, clip]{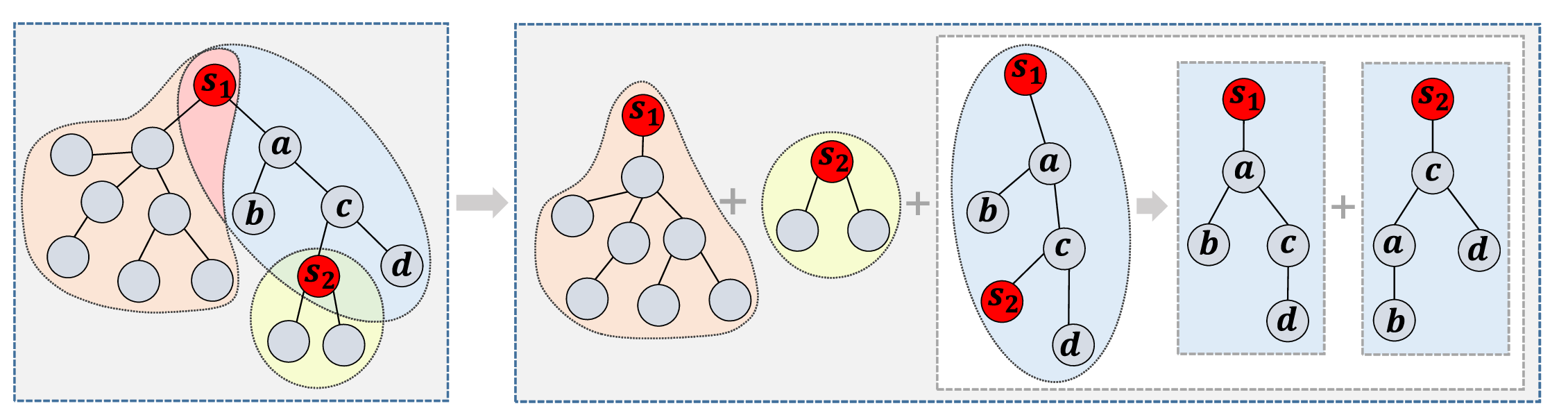}
	\vspace{-6mm}
	\caption{Distant infection sources.}
	\label{fig:multisource}
    \vspace{1mm}
\end{figure}

\begin{figure*}[t]
	\vspace{0mm}
	\captionsetup[subfloat]{captionskip=0pt}	
	\centering
	\subfloat[$n=100$]{%
		\includegraphics[width=0.25\linewidth, trim=0.3mm 0.3mm 0.3mm 0mm, clip]{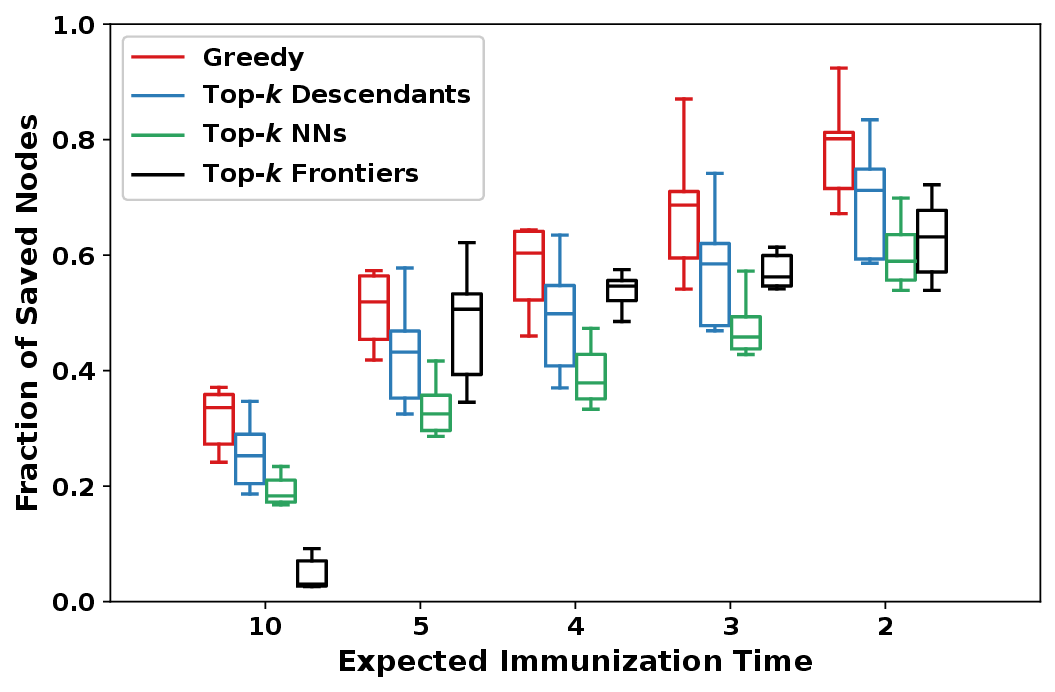}
	}\hspace{4mm}
	\subfloat[$n=500$]{%
		\includegraphics[width=0.25\linewidth, trim=0.3mm 0.3mm 0.3mm 0mm, clip]{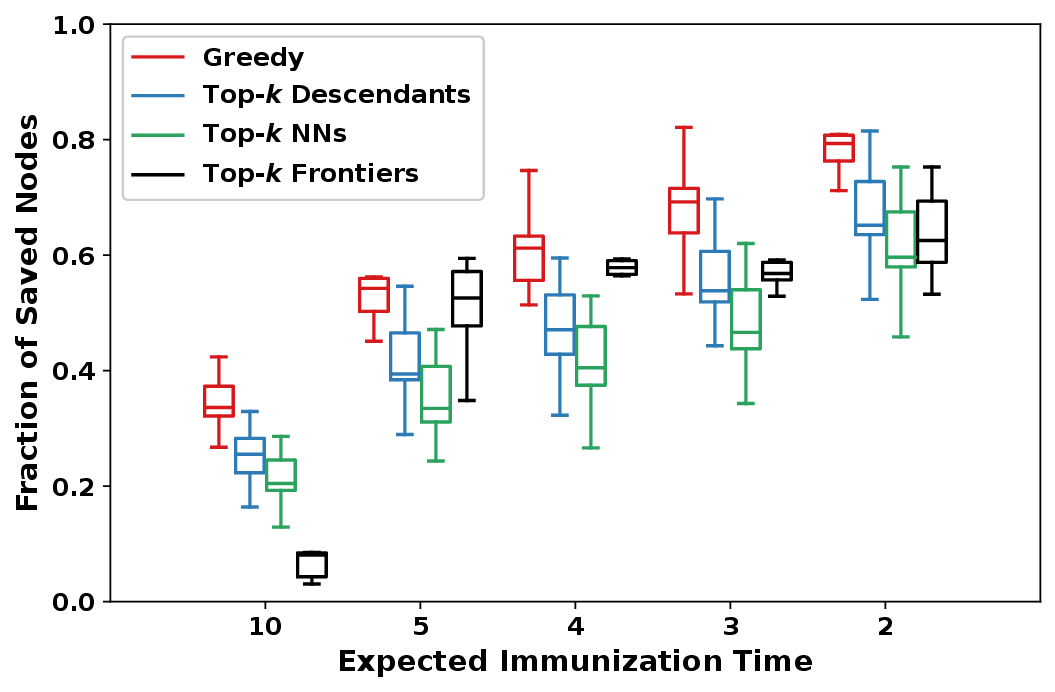}
	}\hspace{4mm}
	\subfloat[$n=1000$]{%
		\includegraphics[width=0.25\linewidth, trim=0.3mm 0.3mm 0.3mm 0mm, clip]{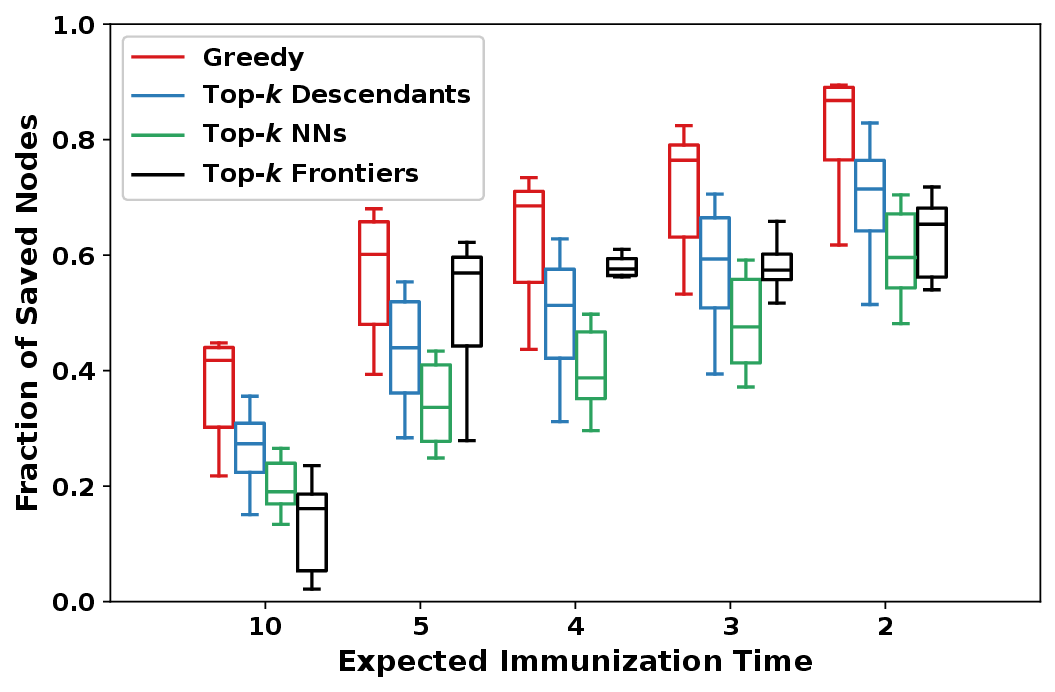}
	}
    \vspace{-1mm}
    \caption{Impact of the expected immunization time $\Ex[\tau]$ under random binary trees with different sizes.}
	\label{fig:s1}
	\vspace{-5mm}
\end{figure*}

\subsection{Distant Multiple Sources}\label{9.3}

The last fundamental case is when there are multiple sources in a tree that are possibly distant from each other while one of them is the root of the tree. In this case, we divide the tree into multiple subtrees as follows. First, for each source of infection, we extract a subtree whose root can be the source while the other nodes are healthy. This process leads to multiple disconnected trees, each of whose root is the source of infection for each tree. During the process, we also replicate the source nodes, i.e., they remain in the tree after the removal of each subtree. Once it is done, the remaining nodes of the tree form a subtree or multiple ones, where healthy nodes can now be infected by multiple sources. See Figure~\ref{fig:multisource} for an illustration. While we can handle the former case with the disconnected trees in the same way as we explained in Section~\ref{9.1}, we below propose an heuristic approach for the latter case.
 
Suppose that we have a tree $G'$ with $m$ multiple sources that we cannot further divide into smaller subtrees in the way as mentioned above. This is the case with the last subtree in Figure~\ref{fig:multisource}, where $m \!=\! 2$. Let $V_{s} \!:=\! \{s_1,s_2,..., s_{m}\}$ be the set of $m$ infection sources. For source $s_i$, we then build a separate \emph{single rooted} tree, say $G'_{s_i}$, having the source as its root while other sources are removed. See Figure~\ref{fig:multisource} (the last process) for an illustration. We assume that the propagation of infection takes place over each tree independently, and a susceptible node is infected as long as it is infected in one of the trees. Using this assumption and inspired by the expression of $r(S)$ in (\ref{reward_function}), we define the following new reward function:
\begin{align}\label{reward_function_special}
	r'(S) := \sum_{i\in S}|N'_i\setminus \underset{{j\in {N'_i\cap S}}}{\cup}N'_j|\prod_{s \in V_s}\pr(Z^s_i>\tau),
\end{align}	
where $N'_i \!:=\! \underset{{s \in V_{s}}}{\cap}N^s_i$, with $N^s_i$ being the set of descendants of node $i$ in $G'_{s}$ (with root $s$), and $Z^s_i$ is the time until node $i$ is infected in $G'_{s}$. Then, we can here use our greedy algorithm in Algorithm~\ref{algo}, but with $r'(S)$ in (\ref{reward_function_special}) instead of $r(S)$ in (\ref{reward_function}). 

\begin{figure*}[t]
	\vspace{0mm}
	\captionsetup[subfloat]{captionskip=0pt}	
	\centering
	\subfloat[$n=100$]{%
		\includegraphics[width=0.25\linewidth, trim=0.3mm 0.3mm 0.3mm 0mm, clip]{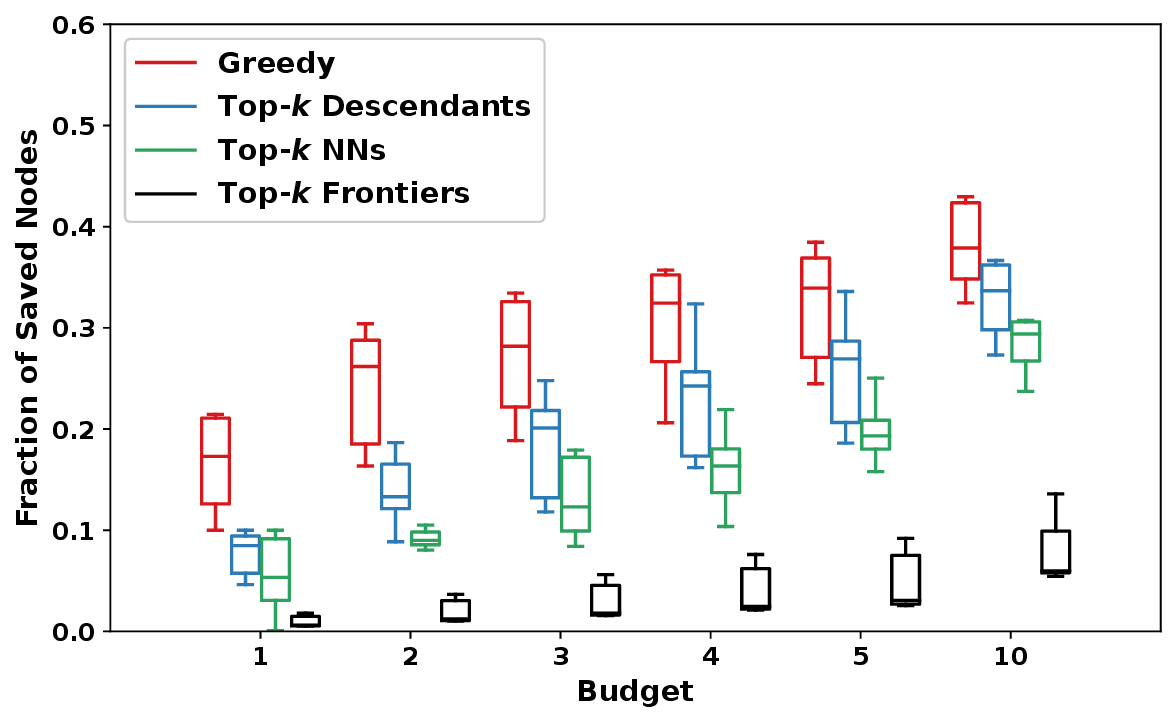}
	}\hspace{4mm}
	\subfloat[$n=500$]{%
		\includegraphics[width=0.25\linewidth, trim=0.3mm 0.3mm 0.3mm 0mm, clip]{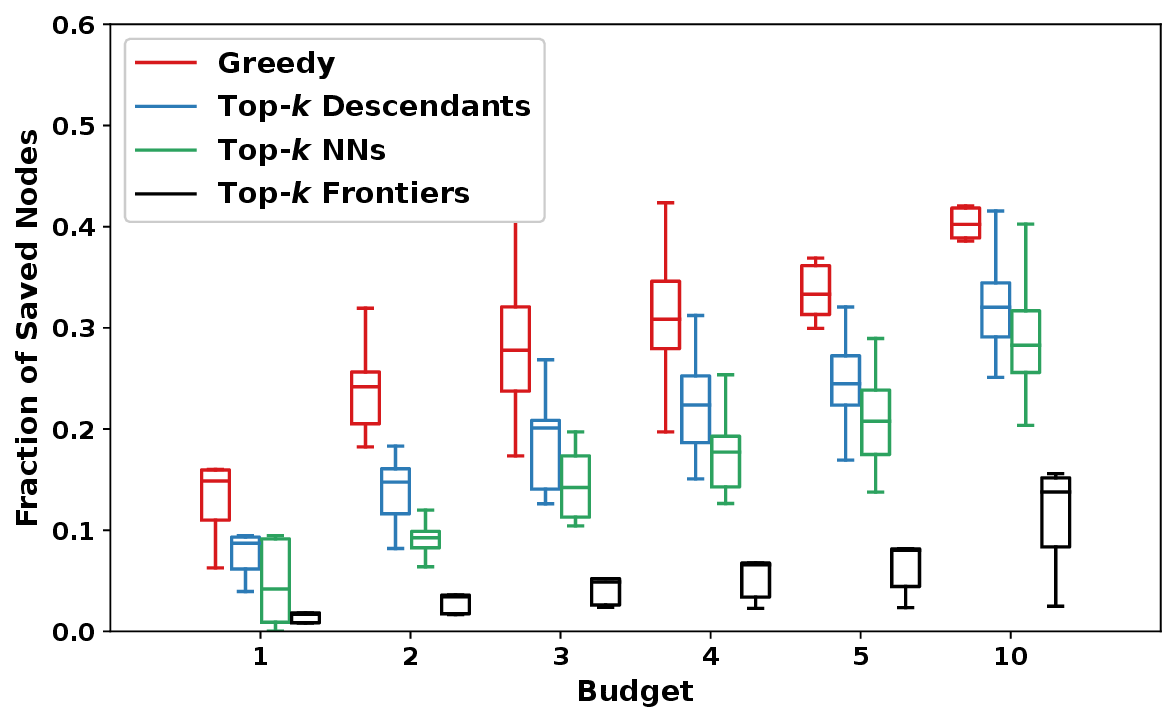}
	}\hspace{4mm}
	\subfloat[$n=1000$]{%
		\includegraphics[width=0.25\linewidth, trim=0.3mm 0.3mm 0.3mm 0mm, clip]{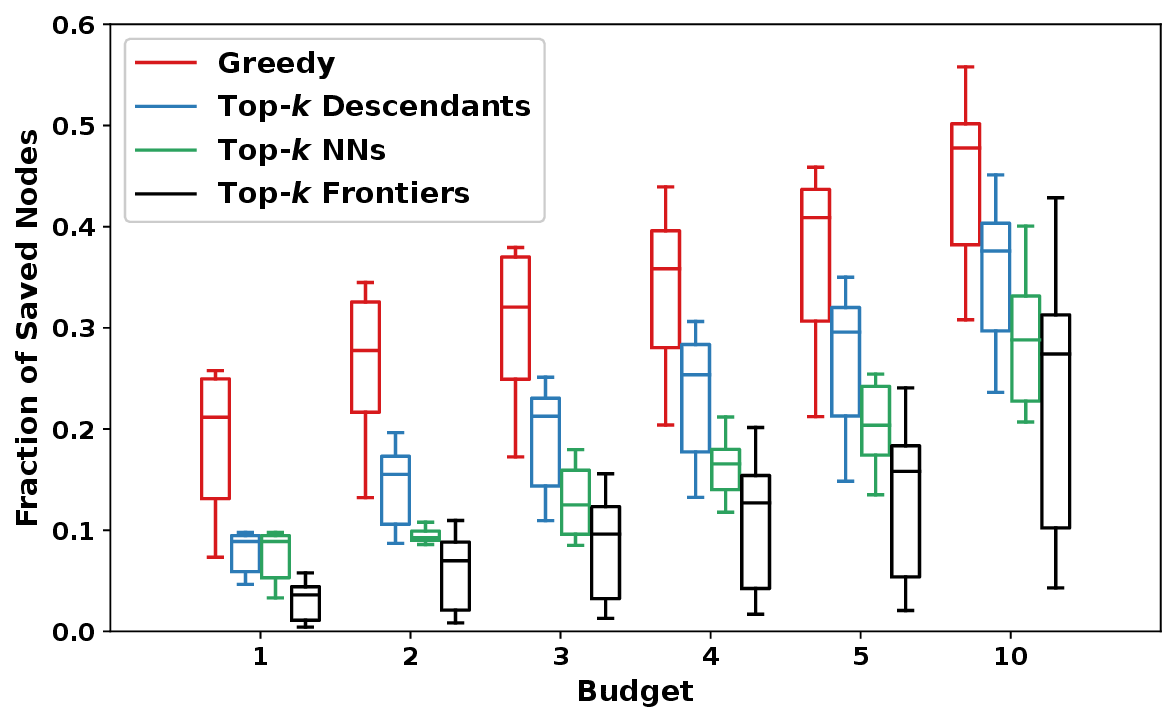}
	}
    \vspace{-1mm}
    \caption{Impact of the vaccination budget $k$ under random binary trees with different sizes.}
	\label{fig:s2}
	\vspace{-6mm}
\end{figure*}

\section{Numerical Simulations}\label{se:simulation}

In this section, we present extensive simulation results to demonstrate the efficacy of our greedy algorithm compared to four baseline vaccination strategies. We focus on general random trees, over each of which an epidemic starts from its root node. Specially, we evaluate the performance of the vaccination policies under random trees with different distributions on the number of children per node and random binary trees.

\subsection{Simulation Setup}

We first explain how we generate random trees and random binary trees. We use the standard Galton-Watson branching process to generate a rooted random tree. Given a root node as the starting point, we connect $\eta$ children to the root, where $\eta$ is an independent random variable with a probability distribution $\pr\{\eta \!=\! i\}$. Assuming that the root has a child or multiple children, we independently connect $\eta$ children to each child of the root. We repeat this process recursively to grow the tree until its size reaches a preset value. If the size of a generated tree does not reach the preset value, we discard the tree and restart the above process to regenerate a new one. For general random trees, we consider that $\eta$ has a Poisson or uniform distribution with mean $\Ex[\eta] \!>\! 0$. For random binary trees, we consider the following distribution for $\eta$: $\pr\{\eta \!=\! 0\} \!=\! 1/6$, $\pr\{\eta \!=\! 1\} \!=\! 1/6$, and $\pr\{\eta \!=\! 2\} \!=\! 4/6$.

The simulation parameters are set up as follows. We set the infection rate $\lambda \!=\! 1$, which means that the expected infection time $1/\lambda$ from a node to its child is a unit of time. We consider that the immunization time $\tau$ is exponentially distributed with rate $\mu$. We choose the values of $\mu$ such that the expected immunization time $\Ex[\tau] \!=\! 1/\mu$ ranges from two units of time to ten units of time. We also vary the value of the vaccination budget $k$ from 1 to 10. The size of each tree, which is either a random tree or a random binary tree, ranges from $n \!=\! 10^2$ to $n \!=\! 10^4$. For each size, we generate ten trees and simulate the epidemic spreading process $10^3$ times on each tree. All reported results are based on the average value over $10^3$ runs. To see the impact of each parameter on the performance of vaccination strategies properly, when we vary the value of a parameter, we set the values of all the other parameters to be the same in the simulations. We use the following default settings, unless otherwise specified: The tree size is $n \!=\! 10^3$, the expected immunization time is $\Ex[\tau] \!=\! 10$ units of time, the vaccination budget $k \!=\! 5$, and the expected number of children of each node in random trees is $\Ex[\eta] \!=\! 3$.

We next explain the baseline vaccination strategies. For the case of random binary trees, we consider three baseline policies to select a set of $k$ nodes to be vaccinated. It is clear that the number of descendants of a node plays an important role in contributing to the expected total reward, as can be seen from (\ref{reward_function}) and our greedy algorithm. We thus consider a baseline policy named `top-$k$ descendants' to select the top $k$ nodes that have the largest numbers of descendants for vaccination. Another baseline policy is to select the top $k$ nodes that are nearest neighbors (NNs) to the source of infection, which is equivalent to selecting the top $k$ nodes with the lowest depth from the root. We name this policy as `top-$k$ NNs'. We also consider a baseline policy named `top-$k$ frontiers' that takes both the immunization time $\tau$ and the number of descendants into account. It selects the top $k$ nodes that have the largest numbers of descendants for vaccination among the nodes which are $\Ex[\tau]$ or more layers away from the source of infection (the root).
For the case of random trees, we consider one more baseline policy besides the above three policies, namely `top-$k$ children'. The intuition is to select the top $k$ nodes with the largest numbers of children as a (parent) node in a random tree can have an unlimited number of children, and they all can be saved upon the immunization of that parent node.

\vspace{-2pt}
\subsection{Simulation Results}
\vspace{-2pt}

We first provide the simulation results for the case of random binary trees. We evaluate the performance of our greedy algorithm and three baseline methods by examining the impacts of the expected immunization time $\Ex[\tau]$, the vaccination budget $k$, and the tree size $n$ on the total number of saved nodes in the end.

We present the fraction of saved nodes (to the tree size) by each vaccination policy with varying expected immunization time in Figure~\ref{fig:s1}. We here report the results when the tree sizes are 100, 500, and 1000. While the fraction of saved nodes increases regardless of vaccination strategies as the expected immunization time $\Ex[\tau]$ decreases, our greedy algorithm outperforms the baseline strategies for all test cases. As shown in Figure~\ref{fig:s1}, the (average) improvements of our greedy algorithm over top-$k$ descendants, top-$k$ NNs, and top-$k$ frontiers are up to 35.6\%, 81.4\%, and 983\%, respectively, which are achieved when $\Ex[\tau] \!=\! 10$ and $n \!=\! 100$. In particular, the top-$k$ frontiers achieves poor performance when $\Ex[\tau] \!=\! 10$, because the nodes whose depths are close to $\Ex[\tau]$ may not have enough descendants and thus contribute little to the final reward.  In addition, we report in Figure~\ref{fig:s2} the fraction of saved nodes when we change the vaccination budget from $k \!=\! 1$ to $k \!=\! 10$, while the expected immunization time remains the same as $\Ex[\tau] \!=\! 10$. The tree sizes are 100, 500, and 1000. Again, our greedy algorithm shows superior performance to the other baseline strategies. We further evaluate the performance of vaccination policies in the fraction of saved nodes when the tree size varies from $n \!=\! 10^2$ to $n \!=\! 10^4$. We report the results in Figure~\ref{fig:s5}, where the vertical line on the top of each bar indicates the 95\% confidence interval. Our greedy algorithm again exhibits superior performance to other policies, regardless of the tree size, with an improvement up to 325\%.

\begin{figure}[t]
	\centering
	\includegraphics[width=0.5\linewidth, trim=1mm 1mm 1mm 0mm, clip]{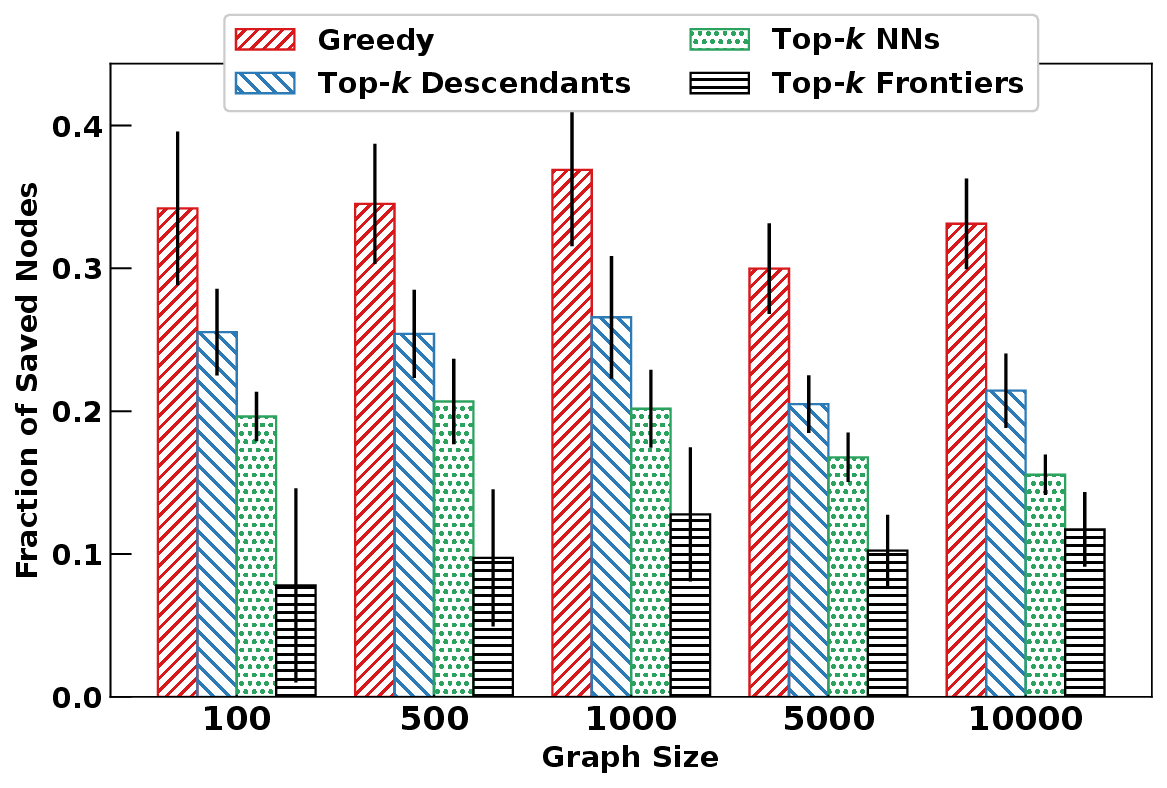}
	\vspace{-3mm}
	\caption{Fraction of saved nodes when the size of a random binary tree changes.}
	\label{fig:s5}
	\vspace{-4mm}
\end{figure}

\begin{figure}[t]
	\captionsetup[subfloat]{captionskip=0pt}	
	\centering
	\subfloat[Uniform]{%
		\includegraphics[width=0.45\linewidth, trim=0.3mm 0.3mm 0.3mm 0mm, clip]{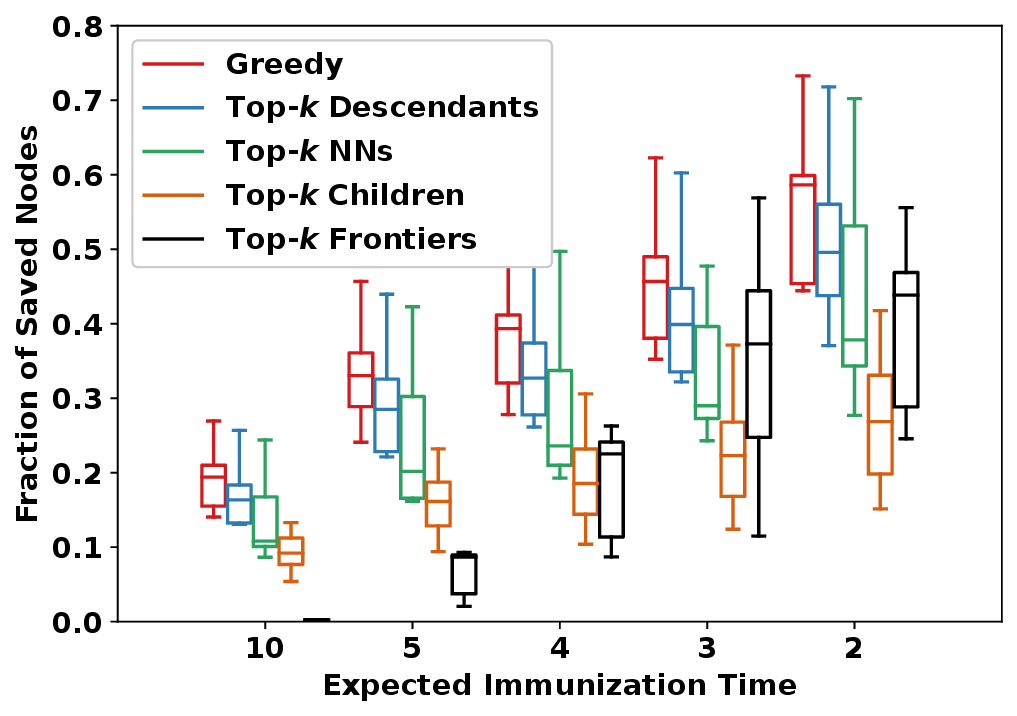}
	}\hspace{0mm}
	\subfloat[Poisson]{%
		\includegraphics[width=0.45\linewidth, trim=0.3mm 0.3mm 0.3mm 0mm, clip]{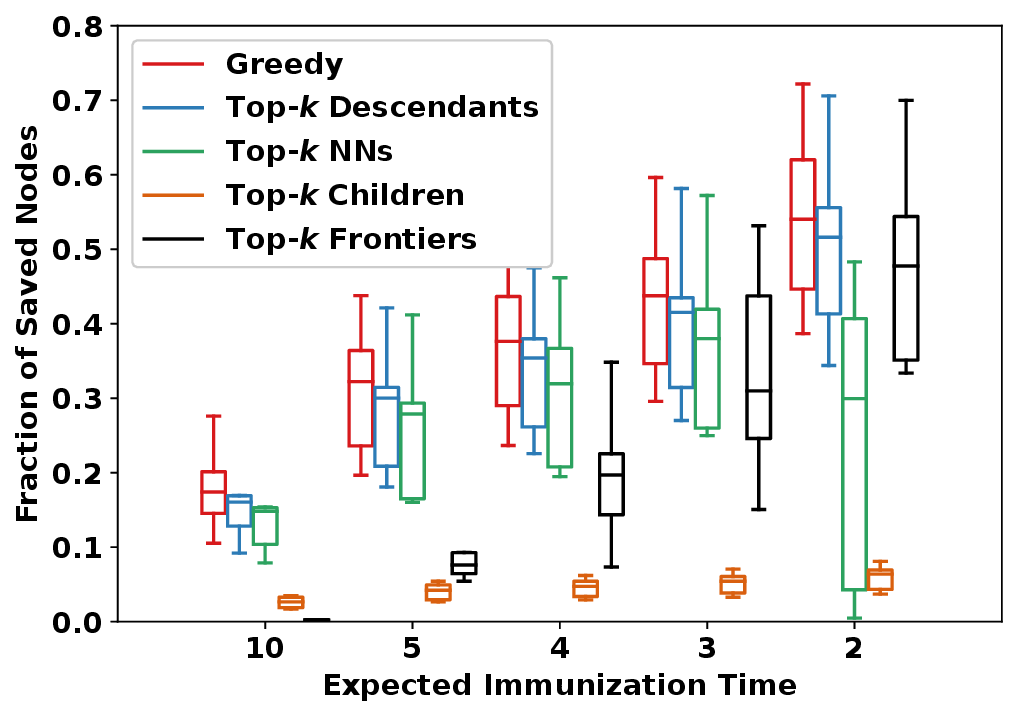}
	}\vspace{-1mm}
	\caption{Fraction of saved nodes with varying expected immunization time $\Ex[\tau]$ under random trees.}
	\label{fig:s3}
\end{figure}

We next turn our attention to the performance evaluation of our greedy algorithm and four baseline policies for the case of random trees. We examine the impacts of the expected immunization time $\Ex[\tau]$, the vaccination budget $k$, the average number of children $\Ex[\eta]$, and the tree size $n$ on the total number of finally saved nodes for each vaccination strategy. As mentioned before, we consider the uniform and Poisson distributions for the number of children $\eta$ per node.

\begin{figure}[t]
	\captionsetup[subfloat]{captionskip=0pt}	
	\centering
	\subfloat[Uniform]{%
		\includegraphics[width=0.45\linewidth, trim=0.3mm 0.3mm 0.3mm 0mm, clip]{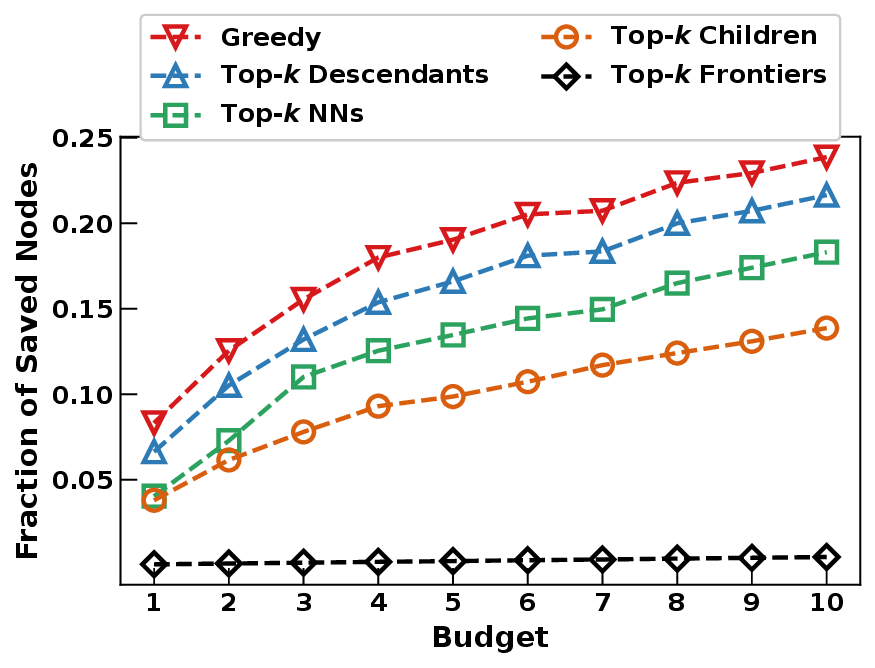}
	}\hspace{0mm}
	\subfloat[Poisson]{%
		\includegraphics[width=0.45\linewidth, trim=0.3mm 0.3mm 0.3mm 0mm, clip]{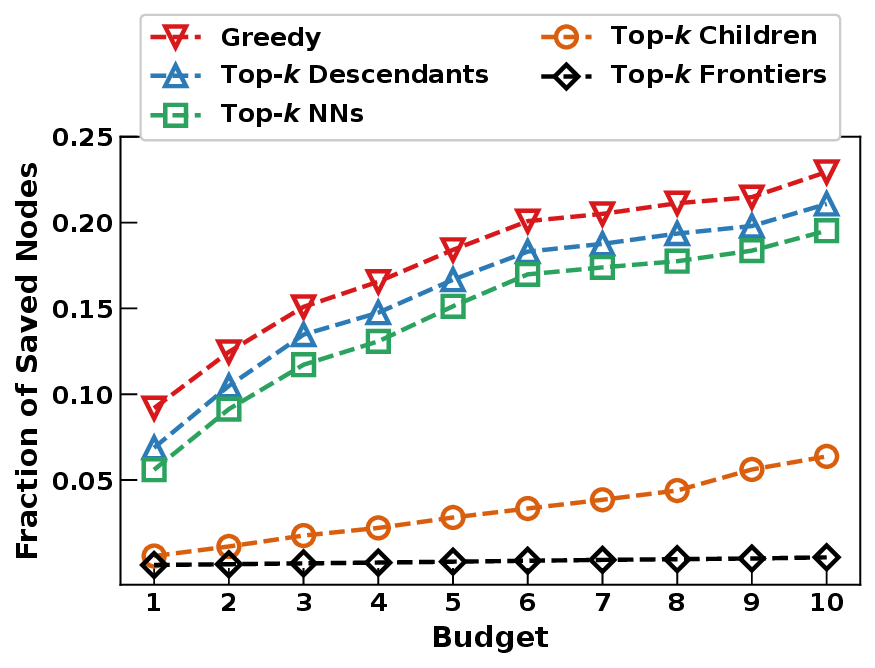}
	}\vspace{-1mm}
	\caption{Fraction of saved nodes with different values of the budget $k$ under random trees.}
	\label{fig:s4}
	\vspace{-3.5mm}
\end{figure}

\begin{figure}[t]
	\captionsetup[subfloat]{captionskip=0pt}	
	\centering
	\subfloat[Uniform]{%
		\includegraphics[width=0.45\linewidth, trim=0.3mm 0.3mm 0.3mm 0mm, clip]{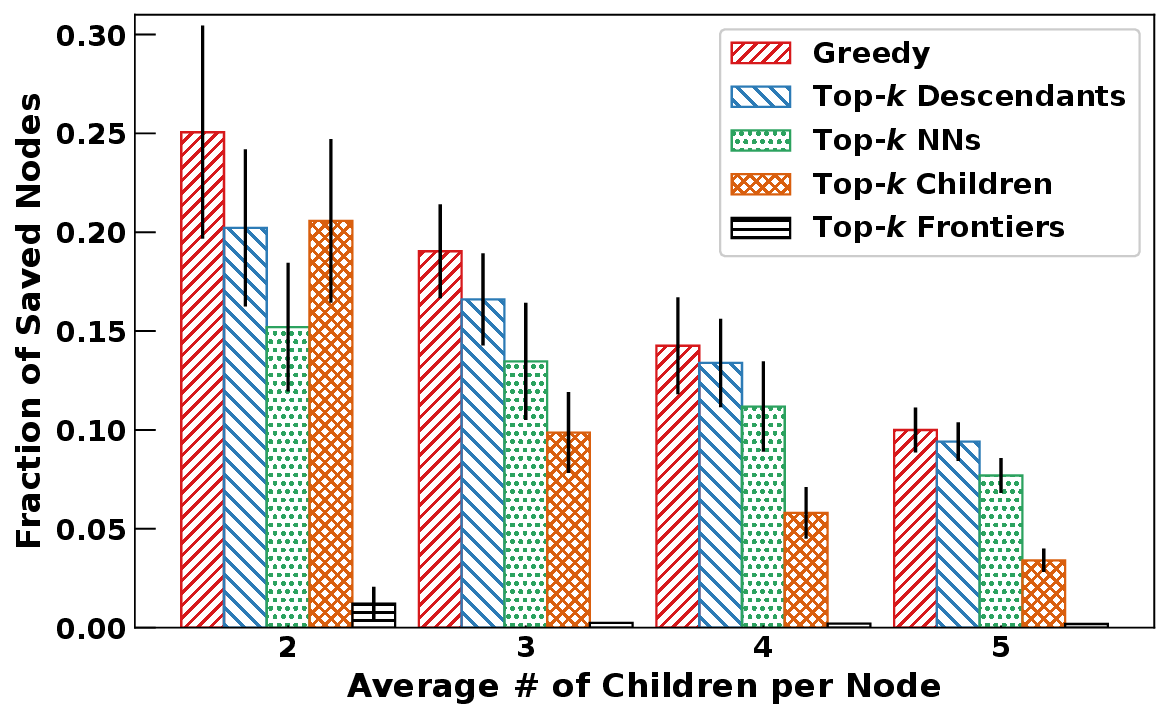}
	}\hspace{0mm}
	\subfloat[Poisson]{%
		\includegraphics[width=0.45\linewidth, trim=0.3mm 0.3mm 0.3mm 0mm, clip]{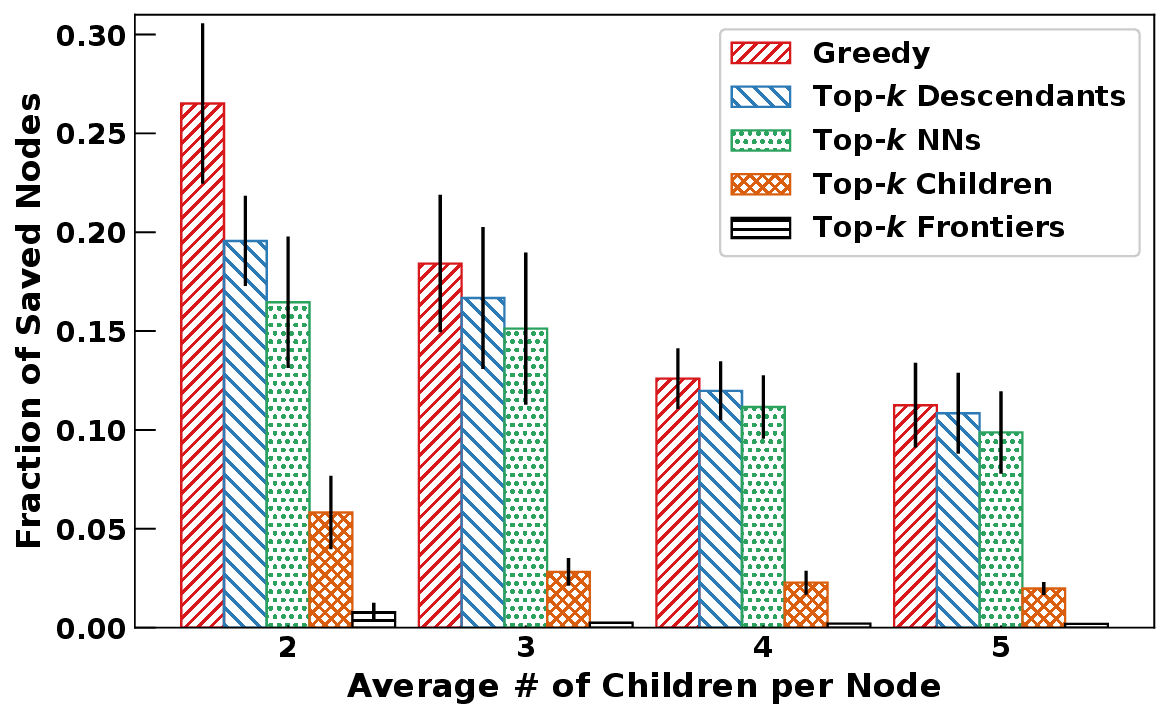}
	}\vspace{-1mm}
	\caption{Average fraction of saved nodes with different choices of $\Ex[\eta]$ under random trees.}
	\label{fig:s6}
\end{figure}

In Figure~\ref{fig:s3}, we present the simulation results to see the impact of the expected immunization time $\Ex[\tau]$ on the total number of saved nodes. The results manifest the superiority of our greedy algorithm over all four baseline policies. We also see that the performance of each vaccination strategy improves as the expected immunization time $\Ex[\tau]$ decreases. Figure~\ref{fig:s4} shows the impact of the budget $k$ on the performance of vaccination policies, while Figure~\ref{fig:s6} indicates how the average number of children $\Ex[\eta]$ affects the performance of each policy. The results again confirm the superiority of our algorithm over the baseline policies. It is also worth noting that the performance of each policy deteriorates as $\Ex[\eta]$ increases, and the performance gap between our greedy algorithm, top-$k$ descendants, and top-$k$ NNs decreases with increasing the value of $\Ex[\eta]$. We observe that when $\Ex[\eta]$ increases, the height of each resulting tree becomes smaller (the tree becomes `wider') as the tree size remains fixed. Therefore, with the same vaccination budget $k \!=\! 5$, the set of nodes for vaccination by each policy becomes overlapped with the ones from other policies.

We finally show in Figure~\ref{fig:s7} how the tree size affects the performance of vaccination policies. Our greedy algorithm again turns out to be substantially better than the baseline policies. We also observe that our algorithm, top-$k$ descendants, and top-$k$ NNs show more or less consistent performance over different tree sizes. The top-$k$ frontiers achieves the worst performance since the small height of random trees limits the number of saved nodes by the policy. While top-$k$ children achieves the second worst performance, its performance gets worse when the tree size increases under random trees with $\eta$ having a Poisson distribution. For large trees, the top-$k$ nodes with the largest numbers of children could be too far from the root node, implying that they could have only few descendants despite their many (direct) children. Thus, the benefit of vaccinating them could be quite limited as the number of saved nodes by them would be mostly limited to their children.

\begin{figure}[t]
	\captionsetup[subfloat]{captionskip=0pt}	
	\centering
	\subfloat[Uniform]{%
		\includegraphics[width=0.45\linewidth, trim=0.3mm 0.3mm 0.3mm 0mm, clip]{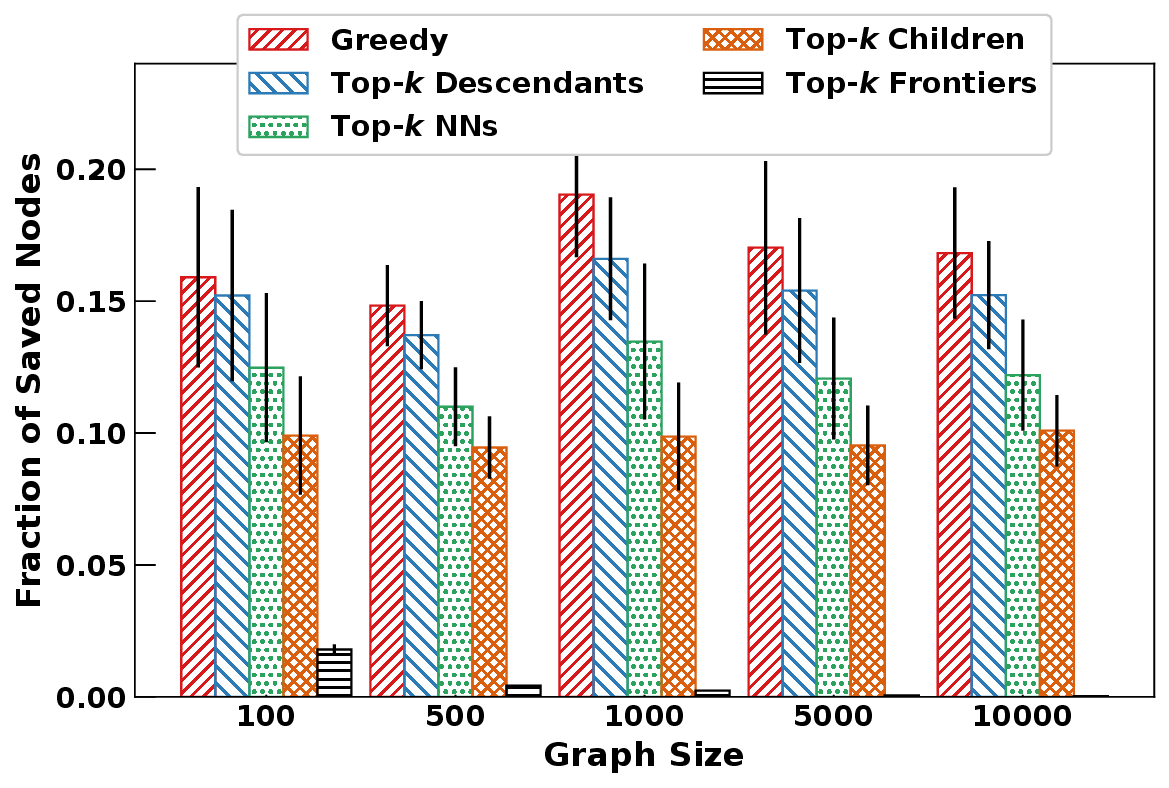}
	}\hspace{0mm}
	\subfloat[Poisson]{%
		\includegraphics[width=0.45\linewidth, trim=0.3mm 0.3mm 0.3mm 0mm, clip]{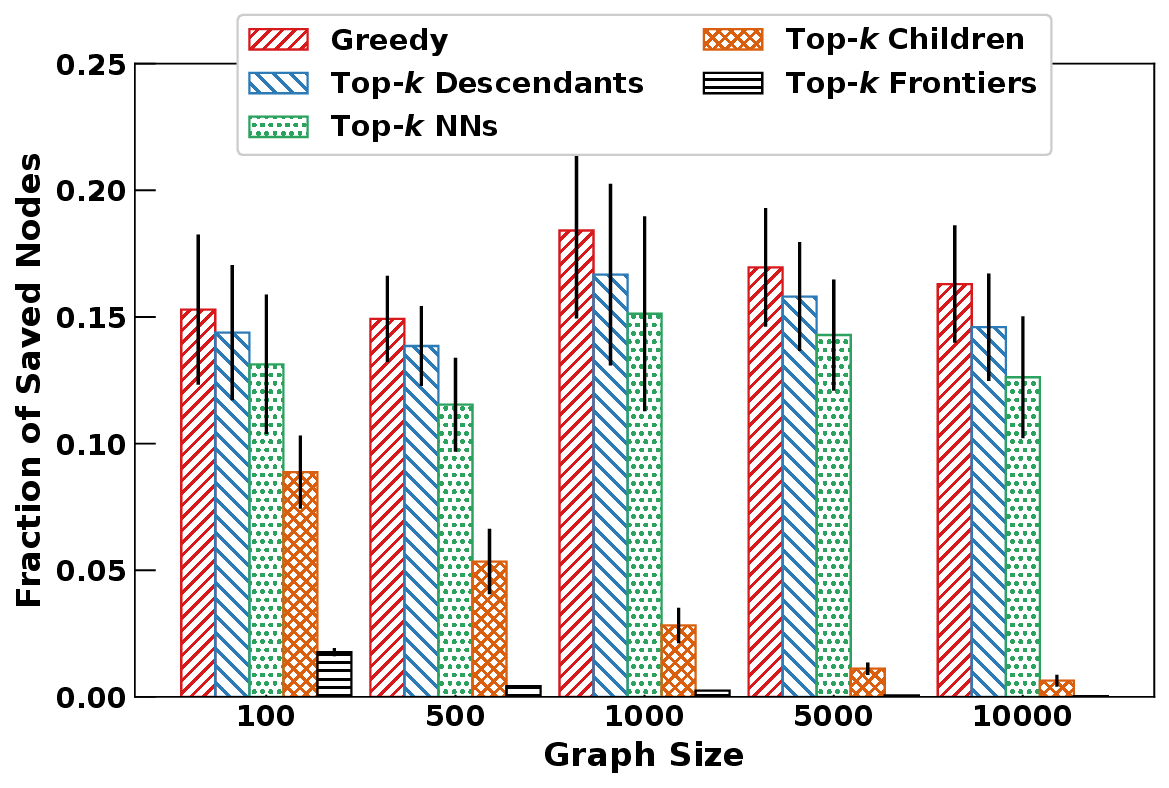}
	}\vspace{-1mm}
	\caption{Average fraction of saved nodes with varying sizes of random trees.}
	\label{fig:s7}
\end{figure}

In summary, the ``top-$k$ children" policy often exhibits bad performance as it is a \emph{source-agnostic} strategy. As explained before, it often ends up choosing a set of nodes for vaccination, which are either infected during the immunization time or have a limited benefit in saving (non-direct) descendants, despite the large numbers of their children. The ``top-$k$ NNs" policy is a \emph{source-aware} strategy, but it is most effective in controlling the propagation of infection at an \emph{early stage} as it chooses the nodes that are $k$ nearest neighbors of the root. Thus, it leads to poor performance when the immunization time/delay is longer. Despite the fact that ``top-$k$ frontiers" is a heuristic policy that seeks for the largest \emph{reward} from vaccination while taking the \emph{expected immunization time} into account, it does not consider the underlying structure of the tree network and fails to capture the essence of the epidemic propagation on the network, thereby leading to unsatisfactory performance. In particular, it exhibits poor performance on random graphs with relatively large average node degrees, and its performance varies significantly over different values of the expected immunization time.

In addition, the ``top-$k$ descendants" policy often achieves better performance than the other baseline policies as the set of nodes chosen for vaccination can save their all descendants as long as they are immune after the immunization time. However, it fails to capture the \emph{risk} of infection of each node during the immunization time. It is also not able to take into account the \emph{overlaps} between the sets of descendants of the target nodes for vaccination, which would be better avoided. In contrast, our greedy algorithm is able to capture both the reward after immunization and the risk of infection during the immunization time. The reward is also computed by taking into account the overlaps between the sets of descendants of the target nodes for vaccination. Hence, it achieves the superior performance over the other policies.

\section{Conclusion}

We have studied the problem of controlling the spread of a viral epidemic in an arbitrary tree network with a limited vaccination budget when vaccines take a non-negligible amount of time to come into effect. We were able to show that the problem is a monotone submodular maximization problem and developed a $(1\!-\!1/e)$-approximation greedy algorithm. We further presented its extension to the scenarios with multiple infection sources. Extensive simulation results under a wide range of scenarios demonstrated the superior performance of our greedy algorithm over other baseline strategies. We believe that our work provides a first step toward the design of vaccination strategies under realistic delay constraints.

\bibliographystyle{IEEEtran}
\bibliography{IEEEabrv,reference}

\begin{thebibliography}{10}
\providecommand{\url}[1]{#1}
\csname url@samestyle\endcsname
\providecommand{\newblock}{\relax}
\providecommand{\bibinfo}[2]{#2}
\providecommand{\BIBentrySTDinterwordspacing}{\spaceskip=0pt\relax}
\providecommand{\BIBentryALTinterwordstretchfactor}{4}
\providecommand{\BIBentryALTinterwordspacing}{\spaceskip=\fontdimen2\font plus
\BIBentryALTinterwordstretchfactor\fontdimen3\font minus
  \fontdimen4\font\relax}
\providecommand{\BIBforeignlanguage}[2]{{%
\expandafter\ifx\csname l@#1\endcsname\relax
\typeout{** WARNING: IEEEtran.bst: No hyphenation pattern has been}%
\typeout{** loaded for the language `#1'. Using the pattern for}%
\typeout{** the default language instead.}%
\else
\language=\csname l@#1\endcsname
\fi
#2}}
\providecommand{\BIBdecl}{\relax}
\BIBdecl

\bibitem{Daley99a}
D.~J. Daley and J.~Gani, \emph{Epidemic Modeling: An Introduction}.\hskip 1em
  plus 0.5em minus 0.4em\relax Cambridge University Press, 1999.

\bibitem{Newman10}
M.~E.~J. Newman, \emph{Networks: An Introduction}.\hskip 1em plus 0.5em minus
  0.4em\relax Oxford University Press, 2010.

\bibitem{Mieghem-ToN09}
P.~Van~Mieghem, J.~Omic, and R.~Kooij, ``Virus spread in networks,''
  \emph{IEEE/ACM ToN}, vol.~17, no.~1, pp. 1--14, 2009.

\bibitem{Nowzari-CS16}
C.~Nowzari, V.~M. Preciado, and G.~J. Pappas, ``Analysis and control of
  epidemics: A survey of spreading processes on complex networks,'' \emph{IEEE
  Control Syst.}, vol.~36, no.~1, pp. 26--46, 2016.

\bibitem{lee2019transient}
C.-H. Lee, S.~Tenneti, and D.~Y. Eun, ``Transient dynamics of epidemic
  spreading and its mitigation on large networks,'' in \emph{ACM MobiHoc},
  2019, pp. 191--200.

\bibitem{li2020trapping}
S.~Li, C.-H. Lee, and D.~Y. Eun, ``Trapping malicious crawlers in social
  networks,'' in \emph{ACM CIKM}, 2020, pp. 775--784.

\bibitem{Chakrabarti08}
D.~Chakrabarti, Y.~Wang, C.~Wang, J.~Leskovec, and C.~Faloutsos, ``Epidemic
  thresholds in real networks,'' \emph{ACM Trans. Priv. Secur.}, vol.~10,
  no.~4, pp. 1--26, 2008.

\bibitem{Ganesh05}
A.~Ganesh, L.~Massouli\'{e}, and D.~Towsley, ``The effect of network topology
  on the spread of epidemics,'' in \emph{IEEE INFOCOM}, 2005.

\bibitem{Draief-AAP08}
M.~Draief, A.~Ganesh, and L.~Massouli\'{e}, ``Thresholds for virus spread on
  networks,'' \emph{Ann. Appl. Probab.}, vol.~18, no.~2, pp. 359--378, 2008.

\bibitem{Prakash-ICDM11}
B.~A. Prakash, D.~Chakrabarti, M.~Faloutsos, N.~Valler, and C.~Faloutsos,
  ``Threshold conditions for arbitrary cascade models on arbitrary networks,''
  in \emph{IEEE ICDM}, 2011, pp. 537--546.

\bibitem{Tong-TKDE16}
C.~Chen, H.~Tong, B.~A. Prakash, C.~E. Tsourakakis, T.~Eliassi-Rad,
  C.~Faloutsos, and D.~H. Chau, ``Node immunization on large graphs: Theory and
  algorithms,'' \emph{IEEE TKDE}, vol.~28, no.~1, pp. 113--126, 2016.

\bibitem{Tong-TKDD16}
C.~Chen, H.~Tong, B.~A. Prakash, T.~Eliassi-Rad, M.~Faloutsos, and
  C.~Faloutsos, ``Eigen-optimization on large graphs by edge manipulation,''
  \emph{ACM TKDD}, vol.~10, no.~4, pp. 1--30, 2016.

\bibitem{Vullikanti-SDM15}
S.~Saha, A.~Adiga, B.~A. Prakash, and A.~K.~S. Vullikanti, ``Approximation
  algorithms for reducing the spectral radius to control epidemic spread,'' in
  \emph{SIAM SDM}, 2015, pp. 568--576.

\bibitem{Preciado-TCNS14}
V.~M. Preciado, M.~Zargham, C.~Enyioha, A.~Jadbabaie, and G.~J. Pappas,
  ``Optimal resource allocation for network protection against spreading
  processes,'' \emph{IEEE Control Netw. Syst.}, vol.~1, no.~1, pp. 99--108,
  2014.

\bibitem{Bilge12}
L.~Bilge and T.~Dumitra\c{s}, ``Before we knew it: An empirical study of
  zero-day attacks in the real world,'' in \emph{ACM CCS}, 2012, pp. 833--844.

\bibitem{Nesara22CSCW}
N.~Dissanayake, M.~Zahedi, A.~Jayatilaka, and A.~Babar, ``Why, how and where of
  delays in software security patch management: An empirical investigation in
  the healthcare sector,'' in \emph{ACM CSCW}, 2022, pp. 1--29.

\bibitem{Nesara22}
N.~Dissanayake, A.~Jayatilaka, M.~Zahedi, and A.~Babar, ``{Software security
  patch management -- A systematic literature review of challenges, approaches,
  tools and practices},'' \emph{Inf. Softw. Technol.}, vol. 144, 2022.

\bibitem{Wang17}
B.~Wang, X.~Li, L.~P. de~Aguiar, D.~S. Menasche, and Z.~Shafiq,
  ``Characterizing and modeling patching practices of industrial control
  systems,'' \emph{Proc. ACM Meas. Anal. Comput. Syst.}, vol.~1, no.~1, pp.
  1--23, 2017.

\bibitem{Google21}
B.~Hawkes, ``{0day ``In the Wild"},'' 2019, \url{https://bit.ly/2w9Ey2e}.

\bibitem{Hacker23}
{The Hacker News}, ``Security navigator research: Some vulnerabilities date
  back to the last millennium,'' 2023, \url{https://bit.ly/3xenCqy}.

\bibitem{cdc21}
M.~W. Tenforde \emph{et~al.}, ``Effectiveness of {Pfizer-BioNTech} and
  {Moderna} vaccines against {COVID-19} among hospitalized adults aged $\geq$65
  years -- {United States, January–March 2021},'' \emph{Morbidity and
  Mortality Weekly Report}, vol.~70, pp. 674--679, 2021.

\bibitem{Shah2011rumor}
D.~Shah and T.~Zaman, ``Rumors in a network: Who's the culprit?'' \emph{IEEE
  Trans. Inf. Theory}, vol.~57, no.~8, pp. 5163--5181, 2011.

\bibitem{shah2012rumor}
------, ``Rumor centrality: A universal source detector,'' in \emph{ACM
  SIGMETRICS}, 2012, pp. 199--210.

\bibitem{luo2013identifying}
W.~Luo, W.~P. Tay, and M.~Leng, ``Identifying infection sources and regions in
  large networks,'' \emph{IEEE Trans. Signal Process.}, vol.~61, no.~11, pp.
  2850--2865, 2013.

\bibitem{ying2016rumor}
K.~Zhu and L.~Ying, ``Information source detection in the {SIR} model: A
  sample-path-based approach,'' \emph{IEEE/ACM ToN}, vol.~24, no.~1, pp.
  408--421, 2016.

\bibitem{choi2019}
J.~Choi, J.~Shin, and Y.~Yi, ``Information source localization with protector
  diffusion in networks,'' \emph{J. Commun. Netw.}, vol.~21, no.~2, pp.
  136--147, 2019.

\bibitem{nemhauser1978analysis}
G.~L. Nemhauser, L.~A. Wolsey, and M.~L. Fisher, ``An analysis of
  approximations for maximizing submodular set functions---{I},''
  \emph{Mathematical Programming}, vol.~14, no.~1, pp. 265--294, 1978.

\end{thebibliography}

\end{document}